\documentclass[unnumsec,webpdf,contemporary,large]{oup-authoring-template}%

\graphicspath{{Fig/}}

\theoremstyle{thmstyleone}%
\newtheorem{theorem}{Theorem}
\newtheorem{proposition}[theorem]{Proposition}%
\theoremstyle{thmstyletwo}%
\newtheorem{remark}{Remark}%
\theoremstyle{thmstylethree}%

\usepackage{graphicx}
\usepackage{bbm}
\usepackage{rotating} 
\usepackage{adjustbox}
\usepackage{amssymb}
\usepackage{multirow} \usepackage{array}
\usepackage{threeparttable}
\usepackage{booktabs}
\usepackage{bbm}

\begin{document}

\journaltitle{Preprint}
\DOI{DOI added during production}
\copyrightyear{2026}
\pubyear{YEAR}
\vol{XX}
\issue{x}
\access{Published: Date added during production}
\appnotes{Paper}

\firstpage{1}


\title[Short Article Title]{Effect-Size Estimation and Confounding-Aware Permutations for PERMANOVA}

\author[1,2,]{Sithija Manage}
\author[2]{Emma Accorsi}
\author[2]{Jacob Nearing}
\author[4]{Siyuan Ma}
\author[2,3]{Eric Franzosa}
\author[2,3,$\ast$]{Curtis Huttenhower}

\address[1]{\orgdiv{Department of Statistics and Data Science}, \orgname{Cornell University}, \orgaddress{\street{Ithaca}, \state{NY 14850}, \country{United States}}}
\address[2]{\orgdiv{Department of Biostatistics}, \orgname{Harvard T.H. Chan School of Public Health}, \orgaddress{\street{Boston}, \state{MA 02115}, \country{United States}}}
\address[3]
{\orgdiv{Harvard Chan Microbiome in Public Health Center}, \orgname{Harvard T.H. Chan School of Public Health}, \orgaddress{\street{Boston}, \state{MA 02115}, \country{United States}}}
\address[4]
{\orgdiv{Department of Biostatistics}, \orgname{Vanderbilt University}, \orgaddress{\street{Nashville}, \state{TN 37203}, \country{United States}}}

\corresp[$\ast$]{Corresponding author. \href{email:email-id.com}{chuttenh@hsph.harvard.edu}}

\received{Date}{0}{Year}
\revised{Date}{0}{Year}
\accepted{Date}{0}{Year}



\abstract{\textbf{Motivation:}
PERMANOVA is widely used to test associations between covariates and high‑dimensional data with arbitrary distributions, such as ecology, metabolomics, or microbial community profiling. In these contexts, its effect size measure (R2) is generally used to indicate how much the high-dimensional measurement covaries with respect to a separate, univariate measure. However, as we show, this conventional effect size is inflated under the null, resulting in values that do not accurately reflect the conventional expectation of "variance explained." Furthermore, standard permutation schemes yield inflated false positives in the presence of non-independent samples, such as repeated measure designs.\\
\textbf{Results:}
We introduce confounding‑aware permutation schemes for both invariant and variant covariates. We additionally develop a causally-grounded, distance‑based effect‑size estimator derived from the Hájek average treatment effect estimator, providing interpretable effect sizes in the units of the chosen dissimilarity metric. Simulations using microbiome profiles demonstrate that these approaches control false positives while maintaining power, outperforming standard PERMANOVA implementations.\\
\textbf{Availability and Implementation:}
An R implementation is freely available in the permaBrioche module of the bioBakery ecosystem {{https://huttenhower.sph.harvard.edu/permabrioche}}.
Contact:
chuttenh@hsph.harvard.edu} 





\maketitle


\section{Introduction}

Permutational multivariate analysis of variance (PERMANOVA) \citep{anderson2001new, anderson2008permanova+} is a popular statistical method for associating covariates of interest with arbitrarily distributed high-dimensional data. It is widely used in fields as diverse as environmental ecology, geology, metabolomics, and microbial community science, all of which frequently yield measurements with such properties. PERMANOVA has thus been adopted in these cases due to its flexibility as a nonparametric, multivariate test that allows use of a distance or dissimilarity metric of choice \citep{anderson2001new}. However, proper implementation and usage of the method is nontrivial in many
common settings, such as repeated-measures or otherwise clustered designs, where
samples from the same subject are not exchangeable and covariates fixed at the
subject level must be tested by permuting whole subjects rather than individual
samples \citep{anderson2003permutation,koh2019distance}. Its behavior is also dependent on the number of samples available relative to the types of covariates. In such cases, the standard effect size measure derived from the estimated coefficient of determination $R^2$ yields unfavorable properties such as inflation under the null and misleading attribution of explained variance. This is frequently the case in large molecular epidemiology studies
\citep{wang2016genome,rothschild2018environment,lloyd2019multi}. While this challenge is well recognized in the methodological literature \citep{kelly2015power}, it has not, to our knowledge, been addressed in most recent studies despite a surge in their employment of PERMANOVA-based methods \citep{kleine2021data}.

Since its introduction, PERMANOVA has inspired various extensions and studies investigating its limitations and applicability, including effect size estimation. The non-centeredness of an unmodified $R^2$ as an effect size (analogous to $\eta^2$ in the parametric ANOVA setting) has been well-observed\citep{AmpliconDataAnalysisOmegaSq, kelly2015power}. Partial $\omega^2$ has been proposed as an unbiased alternative to partial $\eta^2$ for ANOVA \citep{field2005discovering, keppel1991design, olejnik2003generalized}, but this formulation cannot be applied to more complex study designs (e.g. repeated measures), where it overestimates effect sizes. In analyzing compositional differences among ecological communities, ANOSIM \citep{clarke2001change, clarke2008testing} and ECOSIM \citep{gotelli2004ecosim} provide commonly used alternative effect size measures. However, ANOSIMs suffers from small sample size restriction, sensitivity to group homogeneity, and index dependence, while ECOSIM yields variable Type I and Type II errors based on row and column constraints and test index \citep{walters2006comparison, somerfield2021analysis}. A related but distinct inconsistency has been documented for methods that test differential abundance of individual compositional features (e.g. microbial taxa), rather than overall composition: in a comparison of such methods, researchers found that ANOVA \citep{fisher1925smrw}, SIMPER \citep{clarke1993non}, ALDEx2 \citep{fernandes2014unifying}, ANCOM \citep{mandal2015analysis}, and DESeq2 \citep{love2014moderated} often yield incongruent results, implying that biological interpretations are often inconsistent even when applied to the same input abundance data \citep{khomich2021analysing}.

Non-independence of samples can also easily arise in common study designs that utilize PERMANOVA, such as longitudinal studies, other repeated measures, or meta-analyses. In these scenarios, PERMANOVA yields other inaccuracies. Specifically, permutation schemes must
be carefully specified to avoid false positive results; however, it is often unclear how to accurately apply this as an end-user.
For example, suppose researchers are interested in using PERMANOVA to determine if delivery mode in infants is associated with a multivariate outcome in longitudinal data. In this case, subject is a strong confounder as infant delivery mode cannot vary across repeated measurements per subject, and the outcome may exhibit pronounced intra-person consistency and inter-person differentiation \citep{faith2013long, ursell2012interpersonal, zhu2010human}. PERMANOVA calculates a (pseudo-) F statistic that is the ratio of the between-group to the within-group sum of squared distances. The p-value is determined by shuffling the delivery mode labels of the samples and calculating the proportion of F statistics under permutation that are larger than the original F statistic. Under an unadjusted permutation scheme, in which delivery mode labels are randomly re-assigned without regard to the longitudinal structure, the potentially strong subject effects on the outcome are fully attributed to delivery mode. As a result, the null distribution is artificially compressed, yielding false positive associations.

Because the original formulation of PERMANOVA does not address such non-independence, a modest body of literature has emerged to fill this gap. PERMANOVA-S \citep{tang2016permanova} provides a version of confounder adjustment through a regression on residuals combined with Freedman–Lane permutations \citep{freedman1983nonstochastic}, followed by sequential F-tests, firstly partitioning the distance matrix with respect to confounders and secondly with respect to the covariates of interest by fitting linear models to distance matrices. However, for repeated measures their adjustment is limited to within-stratum permutations, which, as demonstrated in Section \ref{sec:results}, is insufficient to address confounding for invariant covariates of interest. \citep{zhu2021constraining} similarly propose including an indicator variable for each set of covariates (compatible with LDM \citep{hu2020testing}) and permuting residuals of Y in the Freedman and Lane scheme, which has admissibility restricted to variant covariates of interest in the repeated measures setting.  The most common computational tool for PERMANOVA, the \texttt{adonis2} function in the  \texttt{vegan} package in R \citep{dixon2003vegan,veganCitationR}, contains two native approaches to adjusting for subject-level confounding - the \texttt{strata} parameter and variable adjustment. However, the former is not usable for complex blocking schemes such as often observed in repeated measures designs and exists largely to maintain backward-compatibility \citep{vegan_discussion600} with older software versions. Moreover, as demonstrated in Section \ref{sec:results}, the latter is unable to properly account for confounding in repeated measures. To address this, we develop a more robust permutation scheme which allows for valid effect size estimation in the setting of repeated-measures samples.

Furthermore, we show that even in the standard, cross-sectional setting, the coefficient of determination calculated in PERMANOVA is inflated under the null hypothesis, with the deviation increasing as the sample size decreases relative to the number of levels in the categorical factor. Thus, we posit that there is a pressing need for interpretable, valid effect size measures for analyzing the association between multivariate outcomes and covariates of interest. Our contribution in this work is twofold: we introduce a causal estimator for interpretable effect size expressed in the units of the chosen distance measure, and we define a new permutation scheme that accounts for sample-level repeated measures that is compatible with both PERMANOVA and our newly introduced effect size estimator. In the process, we also evaluate and compare several distance-based multivariate association methods for repeated-measures designs  (Table \ref{tab:method_comparison}).

 \begin{table*}[htbp]
\centering
\begin{threeparttable}
\caption{Comparison of distance-based multivariate association methods for
repeated-measures designs. $\checkmark$ = supported; $-$ = not supported;
$\circ$ = partial support (see notes).}
\label{tab:method_comparison}
\setlength{\tabcolsep}{4pt}
\renewcommand{\arraystretch}{1.3}
\begin{tabular}{@{}lcccccccc@{}}
\bottomrule
& \multicolumn{2}{c}{\textit{Repeated-measures validity}}
& \multicolumn{3}{c}{\textit{Effect size}}
& \multicolumn{3}{c}{\textit{Inference}} \\
\cmidrule(lr){2-3}\cmidrule(lr){4-6}\cmidrule(lr){7-9}
\textbf{Method}
  & \rotatebox{75}{\parbox{2.4cm}{Valid: invariant\\covariate}}
  & \rotatebox{75}{\parbox{2.4cm}{Valid: variant\\covariate}}
  & \rotatebox{75}{\parbox{2.4cm}{Null-centered\\$\tilde{R}^2$}}
  & \rotatebox{75}{\parbox{2.4cm}{Interpretable\\effect size\tnote{a}}}
  & \rotatebox{75}{\parbox{2.4cm}{Causal / ATE\\interpretation}}
  & \rotatebox{75}{\parbox{2.4cm}{Sequential\\per-term tests\tnote{b}}}
  & \rotatebox{75}{\parbox{2.4cm}{User-specified\\distance metric}}
  & \rotatebox{75}{\parbox{2.4cm}{Multiple distances\\simultaneously}} \\
\midrule
Standard PERMANOVA \citep{anderson2001new}
  & $-$ & $-$ & $-$ & $-$ & $-$ & $\checkmark$ & $\checkmark$ & $-$ \\
\texttt{vegan} G1\tnote{c} \citep{veganCitationR,permute}
  & $-$ & $\checkmark$ & $-$ & $-$ & $-$ & $\checkmark$ & $\checkmark$ & $-$ \\
\texttt{vegan} G2\tnote{d} \citep{veganCitationR,permute}
  & $-$ & $\checkmark$ & $-$ & $-$ & $-$ & $\checkmark$ & $\checkmark$ & $-$ \\
PERMANOVA-S \citep{tang2016permanova}
  & $-$ & $\checkmark$ & $-$ & $-$ & $-$ & $\checkmark$ & $\checkmark$ & $\checkmark$ \\
Zhu et al.\ \citep{zhu2021constraining}
  & $-$ & $\checkmark$ & $-$ & $-$ & $-$ & $\checkmark$ & $\checkmark$ & $-$ \\
(GLMM-)MiRKAT \citep{zhao2015testing,koh2019distance}
  & $\circ$ & $\circ$ & $-$ & $\circ$ & $-$ & $-$ & $\checkmark$ & $\checkmark$ \\
\midrule
\textbf{permaBrioche} (this work)
  & $\checkmark$ & $\checkmark$ & $\checkmark$ & $\checkmark$ & $\checkmark$ & $-$ & $\checkmark$ & $-$ \\
\bottomrule
\end{tabular}
\begin{tablenotes}
\footnotesize
\item[a] Effect size expressed in the units of the chosen dissimilarity metric, beyond
  the unadjusted PERMANOVA $R^2$.
\item[b] Sequential (Type~I) or marginal per-term $p$-values adjusted for other
  covariates in the model (e.g., \texttt{adonis2(..., by = "terms")}). \texttt{permabrioche} currently supports only a joint (omnibus) test of all
right-hand-side covariates, which must be all block-invariant or all
block-variant; sequential and marginal per-term tests (\texttt{by = "terms"} or
\texttt{"margin"}) are not supported.
\item[c] \texttt{adonis2} with \texttt{how(blocks\,=\,subject)}: constrains permutations
  within subjects.
\item[d] 
  First projects out variation attributable to individual subjects using the \texttt{Condition} term, thereby isolating the effect of treatment within subjects. 
 \item[$\circ$] MiRKAT enters the microbiome as a kernel and the trait as the outcome, so it doesn't face the invariant/variant permutation-design problem addressed here. GLMM-MiRKAT (Koh et al., 2019) handles clustering via a mixed model with whole-cluster permutation.
\end{tablenotes}
\end{threeparttable}
\end{table*}


To this end, in Section \ref{sec:methods} we confirm that unadjusted $R^2$, which determines the effect size reported by PERMANOVA, yields inflated estimates under the null when the number of levels $p$ in the categorical variable of interest is large relative to the sample size $n$. We show that an $F$ statistic derived from the proposed adjusted $R^2$, is a monotonically increasing function of the original test statistic $F$ and therefore does not affect $p$-values. In Section \ref{sec:methods:proposedestimator}, we propose an alternative effect size estimator for PERMANOVA in the binary covariate setting based on the causal inference average treatment effect (ATE) Hájek estimator \citep{hajek1971comment}. In Section \ref{sec:methods:permutationschemes}, we introduce our proposed permutation schemes, and in Section \ref{sec:results} we perform simulations with real-world longitudinal microbiome data to explore repeated-measures adjustment for different types of covariates of interest. We demonstrate that the default permutations used in a PERMANOVA test result in excess false positive findings in the presence of repeated measures and explain how permutations in PERMANOVA should be specified to mitigate two different repeated-measures structures. Through this study, we aim to assist in the correct use of PERMANOVA in the presence of confounding, provide an alternative effect size measure which is also compatible with our proposed permutation scheme, and promote detailed, transparent reporting of PERMANOVA-type results. 


\section{Methods}\label{sec:methods}

\subsection{Background}

Permutational Multivariate Analysis of Variance (PERMANOVA) \citep{anderson2001new, anderson2008permanova+} is a semiparametric method used to compare groups of observations and test the null hypothesis $H_0$ that the centroids and dispersions across groups are equivalent. P-values are obtained using a distribution-free permutation technique, and the method allows for estimation of effect sizes, interaction terms, and distribution-free inference. This is in contrast to Multivariate Analysis of Variance (MANOVA), which assumes multivariate normality of the data. We follow \cite{anderson2008permanova+} to formally introduce the method in the one-factor setting.

The remainder of this section is dedicated to defining the test statistic used in PERMANOVA.
Let $\mathbf{Y}$ be an $N\times p$ data matrix consisting of $N$ sampling units and $p$ variables. To conduct a PERMANOVA, we must first define a dissimilarity measure. For example, one such choice that is extensively used by the microbiome community is Bray-Curtis dissimilarity \citep{bray1957ordination}:
\begin{equation}BC_{ij} = 1 - \frac{2C_{ij}}{S_i+S_j} = 1 - \frac{2\sum_{v = 1}^p\min (Y_{iv},Y_{jv})}{\sum_{v=1}^p(Y_{iv}+Y_{jv})} \in [0,1],
\nonumber
\end{equation}

where $Y_{iv}$ and $Y_{jv}$ are the observations for variable $v$ for sampling units $i$ and $j$, respectively. In the shorthand notation, $C_{ij}$ is the sum of the lesser values of each variable. $S_i$ and $S_j$ are the sums across all variables for the respective samples. These distances can be calculated for every pair of sampling units such that we define $\mathbf{D} =\{BC_{ij},\} = \{d_{ij},\}, i = 1,...,N; j = 1,...,N$. 
We then square and center the distance matrix to allow us to calculate the sums of squares partitions without requiring the original data. This matrix, denoted by $\mathbf{G}$, is referred to as ``Gower's matrix" \citep{gower1966some}. Let
$\mathbf{A} = \{a_{ij}\} = \left\{-\frac{1}{2}d_{ij}^2\right\}.$ 
Then 
$\mathbf{G} = \{g_{ij} = a_{ij} - \overline{a}_{i\cdot} - \overline{a}_{\cdot j} + \overline{a}_{..}\},$ where $\overline{a}_{i\cdot} $ is the row-mean of row $i$, $\overline{a}_{\cdot j}$ is the column-mean of column $j$, and $\overline{a}_{..}$ is the overall mean. Assuming there are $g$ groups in our analysis, we define $\mathbf{X}$ as an ($N\times g$) design matrix whose columns are indicators coding groups of an ANOVA factor in a one-way design. The final object for calculating the test statistic is the hat matrix, as derived from the normal equations in linear regression: 
$\mathbf{H = X[X^TX]^{-1}X^T}.$

Finally, we can define the pseudo $F$ statistic in PERMANOVA as 
\begin{equation}F = (SS_A/SS_R)\cdot [(N-g)/(g-1)],\label{eq:pseudoF}\end{equation}

where $SS_T = tr(\mathbf{G}); SS_A = tr(\mathbf{HG});$ $SS_R = tr[(\mathbf{I - H)G}]$ and $SS_T = SS_A + SS_R.$ With sufficient data and permutation iterations, this pseudo $F$ statistic can be used to test the null hypothesis $H_0$ that the multivariate distributions of the groups are the same according to the chosen 
dissimilarity measure. The p-value can be obtained by calculating this pseudo $F$ statistic for random permutations of the $g$ group labels across observations to approximate the null distribution. We refer to this as ``unadjusted" PERMANOVA. \label{sec:defPERMANOVA}

\subsection{Proposed Effect-Size Estimator
}

We characterize some limitations of $R^2$ from PERMANOVA and subsequently introduce and justify our proposed effect‑size estimator. This estimator can be implemented with unadjusted permutations or our proposed permutation schemes introduced in Section \ref{sec:methods:permutationschemes}.

\subsubsection{$R^2$ is inflated under PERMANOVA $H_0$}
\label{sec:methods:r2inflated}

In this section we derive the expectation of the PERMANOVA test statistic under $H_0$ when using Euclidean distance.
We can more directly obtain PERMANOVA partitioning using sums of squared dissimilarities \citep{anderson2014permutational}.
Specifically, 
$SS_T = \sum_{i=1}^{(N-1)}\sum_{j=(i+1)}^Nd_{ij}^2/N,$
where $N$ is the total number of samples, and $d_{ij}$ denotes the pairwise dissimilarities between samples. For the one-way case, $SS_W = \sum_{l=1}^gW_l,$ where $W_l$ is the within-group sum-of-squares for group $l$: $W_l = \sum_{i=1}^{(N-1)}\sum_{j=(i+1)}^N\delta_{ij}^{[l]}d_{ij}^2/n_l,$
where $n_l$ is the sample size for group $l$ and $\delta_{ij}^{[l]}$ is an indicator such that it equals 1 if and only if the samples $i$ and $j$ are both in group $l.$ As before, we utilize $SS_A = SS_T - SS_W$ to obtain the pseudo F statistic in (\ref{eq:pseudoF}). 

Using $R^2 = SS_A/SS_T,$ we observe that \textit{the $R^2$ computed using PERMANOVA with Euclidean distance matches exactly the $R^2$ from a linear model fit to the same data} \citep{anderson2001new}.  For completeness, we have provided a short proof of this equivalence in Appendix \ref{sec:app:R2equivalence}.

Now that we have shown the equivalence between $R^2$ from PERMANOVA with Euclidean distance and $R^2$ obtained from the sums-of-squares decomposition of a linear model with one categorical factor (with $g$ levels), aggregated across all $p$ response variables, we proceed by illustrating that the null distribution of $R^2$ under $H_0$ is upwardly inflated. We consider the Gaussian data case where $Y_i\in \mathbb{R}^p \sim \text{i.i.d. } \mathcal{N}( \mu, \Sigma)$, where $\Sigma = \sigma^2 \mathbf{I}$ for some variance $\sigma^2$.
Under $H_0$, for PERMANOVA with $g$ groups and a multivariate response of dimension $p$, the sums of squares satisfy

\[SS_{A,v} = \chi^2_{g-1}, \qquad SS_{R,v} = \chi^2_{N-g},\] which summed across $p$ independent response dimensions is
$SS_A \sim \chi^2_{p(g-1)}, \text{ and } SS_R \sim \chi^2_{p(N-g)}.$

Thus by taking the ratio of one chi-square variable to the sum of two independent chi-square variables,
\[
R^2_{PERMANOVA} = \frac{SS_A}{SS_A + SS_R} \;\sim\; \text{Beta}\!\left( \frac{p(g-1)}{2}, \frac{p(N-g)}{2} \right)
\]

under $H_0.$ Moreover, the expectation is invariant with respect to $p:$
\[
\mathbb{E}_{H_0}[R^2_{PERMANOVA}] = \frac{p(g-1)}{p(g-1) + p(N-g)} = \frac{g-1}{N-1}.
\]
But the variance decreases as $p$ grows, concentrating the distribution around its mean:
\[
\mathrm{Var}_{H_0}(R^2_{PERMANOVA}) = \frac{(g-1)(N-g)}{(N-1)^2\left(\frac{p(N-1)}{2}+1\right)}.
\]

This result follows from the fundamental \citep{johnson1995continuous} relationship between ratios of independent chi-squared variables and the Beta distribution:
\[
\frac{\chi^2_{a}}{\chi^2_{a} + \chi^2_{b}} \sim \text{Beta}\!\left( \frac{a}{2}, \frac{b}{2} \right),
\]
where in our case $a = p(g-1)$ and $b = p(N-g)$.



This indicates that when PERMANOVA is applied to data where there are truly no group differences, the true effect size is zero, yet the observed $R^2$ will follow a null distribution given by
\[
R^2 \sim \text{Beta}\!\left( \frac{p(g-1)}{2}, \frac{p(N-g)}{2} \right).
\]

Towards illustration, we conducted the following simulation. 
For each of $1000$ replicates, we generated an $N \times p$ data matrix $Y$ with entries
\[
Y_{ij} \sim
\begin{cases}
\mathcal{N}(
\mu = 0,\sigma = 1), & \text{Euclidean case}, \\
\mathrm{Pois}(\lambda = 10), & \text{Bray--Curtis case},
\end{cases}
\qquad \text{independently}.
\]
We considered sample sizes $N \in \{20, 40, 100\}$, number of groups $g = 5$, and $p = 30$ response variables. Group labels were randomly assigned to samples while maintaining equal group sizes. PERMANOVA was then performed using \texttt{adonis2} in $\mathbf{R}$, and the $R^2$ values were extracted to assess differences among groups. Because group labels were assigned at random, the null hypothesis $H_0$ of no group differences holds in all replicates.
Figure~\ref{fig:betaOverlayEuc} depicts the empirical distribution of $R^2$ across the $1{,}000$ simulations, and a similar figure for Bray-Curtis dissimilarity can be found in Appendix \ref{sec:app:R2equivalence} (Fig. \ref{fig:betaOverlayBray}).

\begin{figure}[h!]
    \centering
\includegraphics[width=\linewidth]{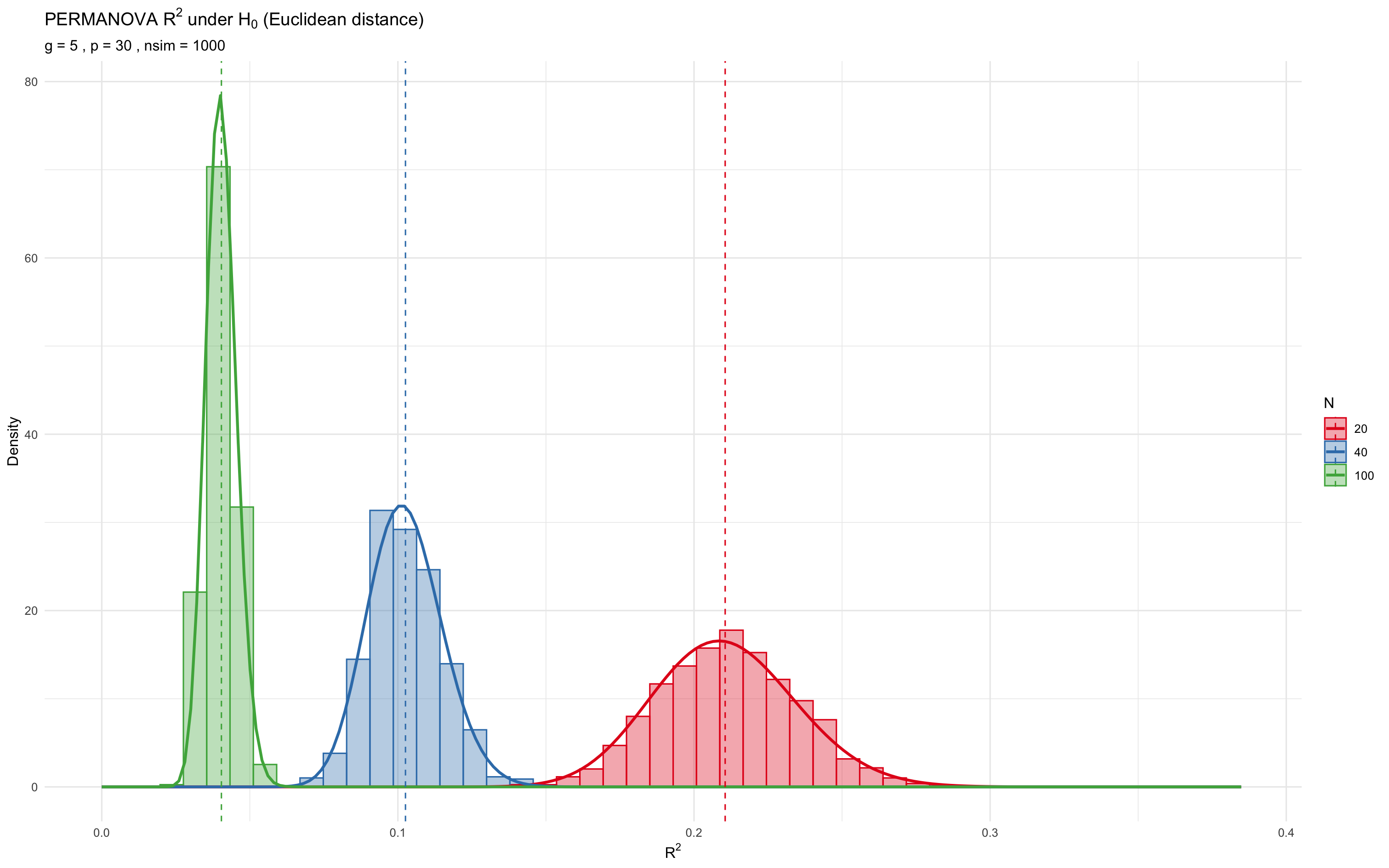}
        \caption{Simulated PERMANOVA $R^2$ distribution under $H_0$, with overlaid $\text{Beta}$ distribution and mean. Red, blue, and green correspond to $N=20$, $N=40$, and $N=100$, respectively. (Euclidean distance)
}
        \label{fig:betaOverlayEuc}
\end{figure}

We observe that the empirical $R^2$ values closely align with the theoretically-expected Beta distribution, highlighting the inflation of $R^2$ from PERMANOVA under $H_0$.

\subsubsection{Parametric null-centered ${R}^2$}
\label{sec:methods:null-centeredR2}

Although bias‑corrected ANOVA effect sizes exist (e.g., $\omega^2$, see Appendix \ref{sec:app:omega2}), at the time of publication, we are not aware of any work that defines a PERMANOVA‑specific effect size by subtracting the exact PERMANOVA $R^2$ null expectation
in the Euclidean setting. To this end, we propose a simple, null-centered adjustment denoted by $\Tilde{R}^2$:

\[\Tilde{R}^2 = R^2 - \frac{g-1}{N-1},\]

where $R^2$ is the unadjusted estimate (obtained from running \texttt{adonis2}, for example). Recall that $g$ is the number of levels in the categorical factor and $N$ is the sample size. Some benefits of using this null‑centered statistic are that a truly zero effect size corresponds to $\tilde R^2=0$, the statistic is centered at zero under random group assignment, and its magnitude reflects the amount of variance explained in \textit{excess} of what would be expected purely from chance partitioning of samples into groups. Even when using a distance other than Euclidean, we can similarly use the estimated null distribution of $R^2$ across permutations to estimate and adjust for the deviation. 
Figures \ref{fig:sub1} and \ref{fig:sub2} in Appendix \ref{sec:app:R2equivalence} empirically validate the following three claims on real metagenomic data from human stool samples, as provided by the \texttt{curatedMetagenomicData} R/Bioconductor package:

\begin{itemize}
    \item Conducting PERMANOVA with Euclidean distance is equivalent to applying a linear model with a single categorical factor in terms of $R^2$.
    \item Unadjusted $R^2$ is upwardly inflated for small sample sizes relative to the number of groups.
    \item Using $\tilde{R}^2$ substantially reduces this deviation from zero.
\end{itemize}

Further details for this analysis can be found in Appendix \ref{sec:app:QMD}.

\subsubsection{Proposed Estimator Definition}
\label{sec:methods:proposedestimator}

Although we address the inflated effect size of $R^2$ in PERMANOVA in the previous section, there still remains the issue of interpretability. PERMANOVA $R^2$ gives a \textit{descriptive} quantification such as ``[$R^2$]\% of the variability in the microbiome is attributable to medication.”
To provide a more concrete interpretation, we propose a novel estimator based in causal inference that provides an interpretable effect size measure for distance‑based multivariate analyses. Specifically, we employ a Hájek-style estimator \citep{hajek1971comment}, for which the details are given in the remainder of this section.

We define our estimator for a binary treatment under the Stable Unit Treatment Value Assumption (SUTVA) framework \citep{rubin1980randomization}, with the aim of supporting causal interpretation. It is important to note, however, that the estimator remains valid even when SUTVA does not hold and a causal interpretation is not justified.
The SUTVA consists of two components: no
interference (i.e., a treatment applied to one unit has no effect on an outcome for any
other units) and no hidden variation of treatment (i.e., for each unit, there are no different forms or versions of each treatment, which lead to different potential outcomes) \citep{imbens2015causal}. 
 We invoke SUTVA to ensure that the average treatment effect estimator is well-defined and interpretable in our analysis.

Furthermore, we assume \emph{unconfoundedness} and \emph{overlap}:  
\begin{itemize}
    \item \textbf{Unconfoundedness (Ignorability):} Conditional on the set of observed covariates $X$, the treatment assignment is independent of the potential outcomes, i.e., $
    \{Y(0), Y(1)\} \;\perp\!\!\!\perp\; T \mid X,$
    where $Y(0), Y(1)\in \mathbb{R}$ denote the potential outcomes under control and treatment, respectively.
    This assumption implies that, after conditioning on $X$, there are no unobserved confounders driving both treatment assignment and outcomes.  
    \item \textbf{Overlap (Positivity):} For all covariate values $X$ in the support, there is a positive probability of receiving each treatment, i.e., $
    0 < P(T=1 \mid X) < 1 .$
    This ensures that units with any given covariate profile could plausibly be observed in both the treated and control groups, allowing for meaningful causal comparisons.  
\end{itemize}

Together with SUTVA, these assumptions provide the conditions under which our Hájek estimator \citep{hajek1971comment}, given in Algorithm \ref{alg:Hajek}, can be interpreted causally.
Let $X_i \in \mathbb{R}^p$ denote the features or profiles for sample $i$ (e.g., a microbiome or omics sample), and $T_i \in \{0,1\}$ indicate the binary treatment or group assignment (e.g., $T_i=1$ for treatment, $T_i=0$ for control). Let $C_i$ be a vector of confounders/covariates for sample $i$.

Fix a user-specified dissimilarity $d(\cdot,\cdot)$ and define the \emph{population} control centroid
$
\mu_0^\star := \mathbb{E}[\,X_i \mid T_i = 0\,] \in \mathbb{R}^p.
$
Define potential distance outcomes $Y_i(t) := d\big(X_i(t),\,\mu_0^\star\big)$ for $t\in\{0,1\}$. Our target parameter is the average treatment effect (ATE)
$
\tau \;:=\; \mathbb{E}\!\left[\,Y_i(1) - Y_i(0)\,\right]
\;=\; \mathbb{E}\!\left[\, d\big(X_i(1),\mu_0^\star\big) - d\big(X_i(0),\mu_0^\star\big) \,\right].
$

A positive $\tau$ indicates that treatment moves samples farther, on average, from the control centroid on the $d$-scale. Furthermore, in the Euclidean distance setting, we derive the decomposition of the effect into a mean component and a variance component: $\tau = \tau_{loc} + \tau_{disp}$.

To estimate $\tau$, we use a \emph{sample-split Hájek} procedure that (i) estimates a control
anchor from training controls, (ii) forms test-set outcomes $Y_i = d(X_i,\hat\mu_{0,\mathrm{train}})$,
(iii) estimates $\hat e_i \approx \Pr(T=1\mid C_i)$ on the training split, and
(iv) computes a normalized IPW (Hájek) ATE on the test split.
This yields an effect size in the \emph{units of $d$}; for squared Euclidean $d$ it admits a
location–dispersion interpretation, while for non-Euclidean $d$ (e.g., Bray–Curtis) the same
anchor-and-weighting logic applies without an exact algebraic decomposition.
Permutation-based inference uses label shuffles that respect the study design (e.g., within- or
between-subject permutations in Section \ref{sec:methods:permutationschemes}) to construct $p$-values and a null-centered effect-size summary.
Full implementation details are given in Algorithm~\ref{alg:Hajek}.

\begin{algorithm}[H]
\caption{Sample-Split Normalized Hájek Estimator}
\begin{algorithmic}[1]
\Require Data $\mathcal{D}=\{(X_i,T_i,C_i,s_i)\}_{i=1}^N$ with features $X_i\in\mathbb{R}^p$,
  treatment $T_i\in\{0,1\}$, numeric covariates $C_i$, subject $s_i$; permutations $B$.
\State Split \emph{subjects} into disjoint training and test sets (default 50/50), redrawing
  until both treatment levels occur in training; let $I_{\text{train}}, I_{\text{test}}$ be the sample indices.
\State Compute the control centroid
  \[\hat\mu_{0,\text{train}} = \frac{1}{|\{j\in I_{\text{train}}:T_j=0\}|}\sum_{j\in I_{\text{train}},T_j=0} X_j.\]
\State For $i\in I_{\text{test}}$, compute $Y_i = d(X_i,\hat\mu_{0,\text{train}})$
  (squared Euclidean or Bray--Curtis).
\State Estimate propensities: if $C$ is supplied, fit the logistic model $T\sim C$ on
  $I_{\text{train}}$ and predict $\hat e_i$ for $i\in I_{\text{test}}$; otherwise set
  $\hat e_i = \frac{1}{|I_{\text{test}}|}\sum_{j\in I_{\text{test}}}T_j$. Truncate $\hat e_i$ to $[10^{-3},\,1-10^{-3}]$.
\State Compute the H\'ajek estimate on $I_{\text{test}}$:
{\color{black}
\[
\hat\tau_{\text{obs}} =
\frac{\sum_{i\in I_{\text{test}}} T_i Y_i/\hat e_i}{\sum_{i\in I_{\text{test}}} T_i/\hat e_i}
\;-\;
\frac{\sum_{i\in I_{\text{test}}} (1-T_i) Y_i/(1-\hat e_i)}{\sum_{i\in I_{\text{test}}} (1-T_i)/(1-\hat e_i)}.
\]
}
\State Compute $\hat\tau_{\text{loc}}$ and $\hat\tau_{\text{disp}}$ on the test set (as before).
\For{$b=1$ to $B$}
  \State Permute $T$ respecting the study design (within-subject for variant, between-subject
    for invariant covariates), keeping the train/test split fixed.
  \State Repeat Steps 2--5 with the permuted labels to obtain $\hat\tau^{(b)}$.
\EndFor
\State $p = \big(1+\sum_{b=1}^B \mathbbm{1}\{|\hat\tau^{(b)}|\ge|\hat\tau_{\text{obs}}|\}\big)/(B+1)$.
\State Compute permutation $z$-score and normalized ratio:
\[
z = \frac{\hat{\tau}_{\mathrm{obs}} - \mu_\tau}{\sigma_0}, \qquad \text{ratio} = \left|\frac{\hat{\tau}_{\mathrm{obs}}}{\mu_\tau}\right|
\]
where $\mu_\tau = \frac{1}{B} \sum_{b=1}^B \hat{\tau}^{(b)}$ and $\sigma_0 = \sqrt{\frac{1}{B} \sum_{b=1}^B \left(\hat{\tau}^{(b)} - \mu_\tau\right)^2}$ are the (estimated) mean and standard deviation of the null distribution $\{\hat{\tau}^{(b)}\}_{b=1}^B$.
\end{algorithmic}
\label{alg:Hajek}
\end{algorithm}

Our proposed estimator has two key advantages; it natively supports \textit{causal} interpretation and yields an easily \textit{interpretable} effect size. For example, our estimator may output an effect size of 55.7, meaning that moving from group 0 to group 1 increases the expected squared Euclidean distance to group 0’s centroid by 55.7." This is illustrated in Figure \ref{fig:hajekDemo} in Appendix \ref{sec:app:R2equivalence}. Furthermore, the outputted effect size admits a natural decomposition into a location component and a dispersion component. When $Y_i$ is computed using Bray–Curtis distance, $\hat{\tau}_{loc}$ and $\hat{\tau}_{disp}$ remain interpretable Euclidean feature-space descriptors but do not algebraically sum to $\hat{\tau}$. Thus, they still provide a notion of location and dispersion effects even though they are not a decomposition of the original effect size.
A short proof of the location and dispersion decomposition (including with estimated propensity scores) is given in Appendix \ref{sec:app:locationDispersion}.

\subsection{Permutation schemes to address repeated measures}

\label{sec:methods:permutationschemes}
We consider two settings for the binary covariate of interest - ``invariant" and ``variant." An invariant attribute remains fixed across all repeated measures for a subject, while a variant covariate may assume different values across those measures (Figure \ref{fig:emmaPermutationSchemeDiagram}). We provide a proof of the exactness of our permutation scheme in Appendix \ref{sec:app:omnibus_validity}.

For the invariant covariate of interest, we compared the following four methods: ``unadjusted" PERMANOVA, ``Var+Subj" PERMANOVA, PERMANOVA with our proposed "between block" permutations, and our proposed Hájek estimator with ``between block" permutations. Naive PERMANOVA was implemented as the “unadjusted” permutation scheme, in which all covariate labels are permuted (shuffled) at random with no regard for subject information: $
\textbf{``Unadjusted": } \texttt{adonis2(distances$\sim$variant,...)}$.
This was implemented using a call to the \texttt{adonis} function with only the covariate of interest included, which is described in Section \ref{sec:defPERMANOVA}. Second, we ran an implementation of \text{adonis2} in which subject is a variable in the statistical test with it following the covariate of interest. Third, our subject-aware between block permutation scheme implements Algorithm \ref{alg:customInvariant} to specify the permutations and shuffled covariate labels by randomly permuting the covariate at the subject level and then assigning the new covariate variable to all time points for the subjects. Finally, we also ran our proposed Hájek effect-size estimator using the subject-aware between block permutation scheme.

\begin{algorithm}[H]
\caption{Repeated Measures PERMANOVA with Block-Wise Permutations}
\begin{algorithmic}[1]
\Require Distance object \texttt{D}; formula \texttt{D $\sim$ X$_1$ + $\dots$ + X$_k$}; metadata \texttt{data};
  blocking variable \texttt{blocks}; number of permutations $B$; optional \texttt{center\_R2}
\Ensure \texttt{adonis2} omnibus table with block-aware empirical $p$-value
\State Align \texttt{data} to the labels of \texttt{D}; if \texttt{na.rm}, drop samples with missing
  metadata in a block-aware manner, otherwise stop on any \texttt{NA}
\State Classify each covariate as \emph{invariant} (constant within every block) or \emph{variant};
  stop if both types are present
\State Store variant covariates in \texttt{permute\_within} (one row per sample) and invariant
  covariates in \texttt{block\_data} (one row per block)
\State Fit \texttt{adonis2(D $\sim$ .)} with \texttt{by = NULL} and \texttt{permutations = 0} on the observed
  metadata; record the omnibus $R^2_{\text{obs}}$
\For{$b = 1$ to $B$} \Comment{Run null permutations}
    \State Shuffle rows of \texttt{permute\_within} within blocks via \texttt{shuffle(..., how(blocks = blocks))}
    \State Shuffle rows of \texttt{block\_data} across blocks and reassign to samples by block membership
    \State Fit \texttt{adonis2(D $\sim$ .)} on the permuted metadata; record $R^{2}_{(b)}$
\EndFor
\State Compute the empirical $p$-value
\[
p = \frac{1 + \sum_{b=1}^{B} \mathbbm{1}\left\{R^2_{(b)} > R^2_{\text{obs}}\right\}}{B + 1}
\]
\If{\texttt{center\_R2}}
    \State $\tilde R^2_{\text{emp}} \gets R^2_{\text{obs}} - \frac{1}{B}\sum_{b=1}^{B} R^2_{(b)}$
\EndIf
\State \Return \texttt{adonis2} table with $R^2_{\text{obs}}$, $p$ (and $\tilde R^2_{\text{emp}}$ if requested)
\end{algorithmic}
\label{alg:customInvariant}
\end{algorithm}

For the variant covariate of interest, seven permutation schemes were compared. First, we included the “unadjusted” permutation scheme defined in Section \ref{sec:defPERMANOVA} as a benchmark. Second, our “within block” permutation scheme specifies the permutations and randomly reassign covariate values across time points within a subject, as illustrated in Figure \ref{fig:emmaPermutationSchemeDiagram}. Third, we ran our proposed Hájek effect-size estimator with constant propensity score, using the subject-aware within block permutation scheme.

To introduce the remaining methods, we first highlight the fact that PERMANOVA is fit sequentially. In a call to the \texttt{adonis2} function as \texttt{adonis2(distances $\sim$ A + B)}, the variation explained by \texttt{A} is attributed to \texttt{A}, and then the variation explained by \texttt{B} after accounting for \texttt{A} is attributed to \texttt{B}. This impacts the projection matrix, $\mathbf{H}$ used for each term, which may yield markedly different p-values than an implementation \texttt{adonis2(distances $\sim$ B + A)}. We included two additional implementations using the \texttt{adonis} function in the R \texttt{vegan} package: specifying subject as a variable in the statistical test with it preceding the covariate of interest, and specifying subject as a variable in the statistical test with it following the covariate of interest:
$\textbf{``$\sim$Subj+Var": } \texttt{adonis2(distances$\sim$subject+variant,...)}\\
\textbf{``$\sim$Var+Subj": } 
\texttt{adonis2(distances$\sim$variant+subject,...)}.$

Finally, we implemented two \texttt{vegan} developer-suggested techniques to handle confounding with the \texttt{permute} R package \citep{vegan_discussion600, vegan_discussion604, permute}. The first solution specifies blocking by subject (samples within the same subject are not permuted across subjects): 

$
    \textbf{``G1": } \texttt{h1 <- with(data, how(..., blocks = subject))}\\
\texttt{adonis2(distance $\sim$ variant, data, permutations = h1)}$

The second solution applies redundancy analysis with a conditional term for subject ID in order to partial out the variation due to subject:

$
\textbf{``G2": }\texttt{ord<-rda(response\_df$\sim$variant + Condition(subject), data)}\\
    \texttt{anova(ord, permutations = h1)}\nonumber.
$

This second approach first projects out variation attributable to individual subjects using the \texttt{Condition} term, thereby isolating the effect of treatment within subjects. This is combined with subject-level blocking in the permutation scheme. It is also of note that, according to the developers of \texttt{vegan}, \texttt{strata} is a legacy argument of \texttt{adonis2} and should be superseded by the \texttt{permute} package to specify permutations for repeated measures \citep{vegan_discussion600}.
Within iterations of the simulation, each of the four (invariant) and seven (variant) correction methods used the same random seed to improve comparability.

\subsection{Numerical Experiments}

\subsubsection{Data}
\label{sec:methods:data}
We used publicly available stool metagenomes from the Integrative Human
Microbiome Project (HMP2; \citealp{integrative2019integrative}) Inflammatory Bowel
Disease Multi'omics Database \citep{lloyd2019multi}, a longitudinal study in which participants with Crohn's
disease, ulcerative colitis, and non-IBD controls were sampled repeatedly over
approximately one year. Its repeated-measures structure, with many samples per
subject, makes it a natural testbed for block-aware permutation methods.
Species-level relative abundances were extracted from the MetaPhlAn taxonomic
profile table, transposed to a samples-by-species matrix, cleaned to harmonize
sample identifiers, converted to numeric format, and filtered to remove samples
or taxa with all-zero abundance. To reduce computational burden, we randomly
selected 250 samples and retained the 200 most prevalent species, ranked by the
number of samples in which each species was present (with mean abundance as a
tiebreaker). Metadata were reduced to subject and sample identifiers,
deduplicated, cleaned to match the taxonomic profiles, and aligned one-to-one
with the filtered abundance matrix. These real abundance profiles and subject
memberships were retained in all simulations, with covariates of interest and
effects generated synthetically as described below.
Each simulation setting was replicated 500 times.

The same analysis was also conducted on nasal metagenomes from a longitudinal
mother--infant cohort study of \textit{Staphylococcus aureus} carriage
(208 samples from 36 mothers and 36 infants, sampled monthly over the first year
of life; \citealp{accorsi2020determinants}), using daycare status and delivery mode for the variant and invariant blocks, respectively.
Results are in Appendix~\ref{sec:app:staph}.




\subsubsection{Generation of ``variant" and ``invariant" covariates of interest}

\label{sec:generationOfInvVar}
To simulate an invariant covariate of interest that was completely linked with subject (i.e., sex, delivery mode), we randomly generated a Bernoulli variable per subject and assigned the variable to all samples from that subject: $
Z_i \sim \text{Bernoulli}(p),\, \text{for } i = 1, \dots, N \\
X_{ik} = Z_i, \, \text{for } k = 1, \dots, g$.

Here, \( Z_i \) is a subject-level covariate, and \( X_{ij} \) is the covariate value assigned to sample \( k \) of subject \( i \). This defines a covariate that is invariant within each subject. The Bernoulli probability parameter $p$ was varied across simulations, taking values $0.1, 0.25,$ and $0.5$.

To simulate a variant covariate of interest that was partially linked with subject (i.e., daycare attendance or breastfeeding in the preceding month), we randomly assigned subjects to be “type 1” or “type 2” with probability $0.5$: $
T_i \sim \text{Bernoulli}(0.5), \quad \text{for } i = 1, \dots, N,$
where \( T_i = 1 \) denotes a ``type 1'' subject and \( T_i = 0 \) denotes a ``type 2'' subject.
Next, we randomly generated a number of binary values for the “variant” covariate of interest equal to the number of time points for the subject with probability $p$ for “type 1” individuals and probability $1-p$ for “type 2” individuals:
\[
X_{ik} \sim \text{Bernoulli}(q_i), \quad \text{where } 
q_i = 
\begin{cases}
p, & \text{if } T_i = 1 \ (\text{type 1}) \\
1 - p, & \text{if } T_i = 0 \ (\text{type 2}),
\end{cases}
\]

for $i=1,..,N$, $k = 1,...,g$. The Bernoulli probability parameter $p$ was varied across simulations, taking values $0.1, 0.25$, and $0.5$. In the case of $p = 0.5$, the covariate of interest is independent of subject type, and thus we expect zero effect under the unadjusted setup.
Figure \ref{fig:emmaPermutationSchemeDiagram} provides a unifying diagram illustrating the permutation schemes as well as difference between invariant and variant covariates. 

\begin{figure}
    \centering
    \includegraphics[width=0.9\linewidth]{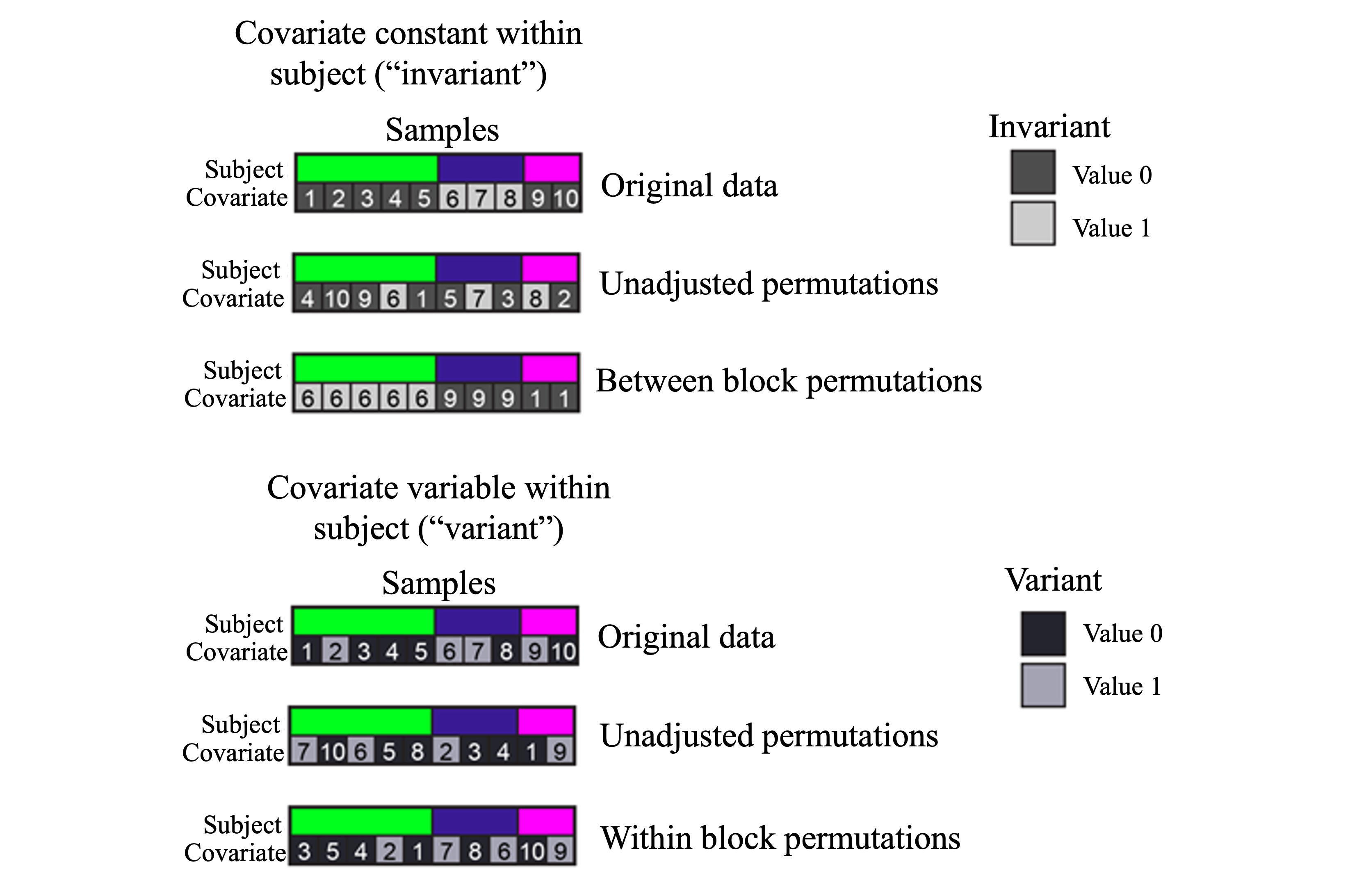}
    \caption{Permutation schemes proposed for confounding correction}
\label{fig:emmaPermutationSchemeDiagram}
\end{figure}

\subsubsection{Signal Strength Manipulation and Explanation of Methods}

In the invariant setting, in order to vary the signal strength in this simulation, we implemented a k-means clustering-based method. The centroid of each subject was used as input into a k-means clustering algorithm with $2$ centers (and $50$ initializations). Subjects in the group with the larger number of observations were assigned $T_i = 0$. Effect size was modulated by ``flipping" each subject with probability $p\in \{0,0.05,0.1,0.5\}.$ This probability is denoted as ``non-independence measure," and it should be noted that probability $0.5$ corresponds to the null signal case under which p-values should have a Unif$(0,1)$ distribution.

Under the variant setting, $10\%$ of the microbial taxa (chosen uniformly at random) received a multiplicative spike in relative abundance among treated samples, i.e., for each selected taxon we multiplied its relative abundance in subjects with $T_i = 1$ by a factor of $(1 + \text{effect\_size})$. Following this, subject-level confounding was introduced by adding an independent taxon-specific shift drawn from $\mathrm{Unif}(0, 0.1)$ to all samples within each subject, followed by renormalization of each sample's composition to sum to $1$.
Throughout this work, Bray-Curtis dissimilarity was used across PERMANOVA implementations.

\subsection*{Software}

All methods are implemented in the open-source \texttt{permabrioche} R package, available as part of the \texttt{bioBakery} ecosystem (GitHub: \texttt{biobakery/permabrioche}). 

The core user-facing functions are:
\begin{enumerate}
  \item \textbf{\texttt{PERMANOVA\_repeat\_measures()}} \\
  Performs PERMANOVA while respecting repeated-measures designs. Users supply a distance object and a model formula (e.g.\ \texttt{D \textasciitilde{} X1 + X2 + ...}), along with a blocking variable identifying groups. The function automatically separates within-block and block-level covariates, constructs between and within-block permutations, and returns a standard \texttt{adonis2} table with empirical p-values based on the customized null.

  \item \textbf{\texttt{hajek\_repeat\_measures()}} \\
  Implements the sample-split normalized H\'ajek estimator on distances (Algorithm \ref{alg:Hajek}) to a control centroid for a binary covariate of interest. The function takes a distance-style formula, sample-level metadata, a feature matrix (e.g.\ taxa abundances), a blocking variable, and the name of the binary treatment covariate. Optional arguments allow the user to choose the distance metric for the outcome (\texttt{"bray"} or \texttt{"euclidean"}), specify the train/test split ratio, and, when desired, provide numeric covariates for propensity-score estimation.

  \item \textbf{\texttt{hajek\_repeat\_measures\_loc\_and\_disp()}} \\
  Computes the location and dispersion components of the H\'ajek effect size under a subject-level split, using the same input structure as \texttt{hajek\_repeat\_measures()}. This function returns the total effect on the distance scale together with its decomposition into a centroid-shift (location) term and a within-group dispersion term.
\end{enumerate}

\section{Results}\label{sec:results}

In order to characterize and improve PERMANOVA effect-size measures and significance calculations, we carried out two families of experiments. In the first, we evaluated effect-size measures in isolation (unadjusted $R^2$,
the parametric null-centered $\tilde{R}^2$, an empirically null-centered $R^2$
obtained from the permutation null, and the H\'ajek ATE estimator) using original and shuffled data drawn from the Human Microbiome Project 2 Inflammatory Bowel Disease Multi-omics Database population \citep{lloyd2019multi}. We subsequently assessed p-value calculations in data with sample-level non-independence and evaluated permutation schemes for block-invariant covariates (such as subject or location across a time course) and block-variant covariates (e.g. diet or temperature).

\subsection{PERMANOVA-based effect sizes are inaccurately reported by unadjusted R2}

Before addressing non-independence, we first evaluated candidate effect-size
measures in a setting with strictly independent samples, retaining a single
sample per HMP2 subject and drawing $N \in \{20, 40, 80\}$ samples without
replacement in each of 500 iterations. Samples were randomly assigned to
$g \in \{2, 5\}$ balanced groups, so that the null hypothesis of no group
difference held exactly. We compared four measures: the unadjusted $R^2$
reported by \texttt{adonis2}; the parametric null-centered
$\tilde{R}^2 = R^2 - (g-1)/(N-1)$ (Section~\ref{sec:methods:null-centeredR2}); an
empirically null-centered $R^2$, obtained by subtracting the mean $R^2$ across
permutations; and the H\'ajek ATE estimator (Algorithm~\ref{alg:Hajek}), which
applies only to binary covariates ($g = 2$).

\begin{figure}
    \centering
    \includegraphics[width=\linewidth]{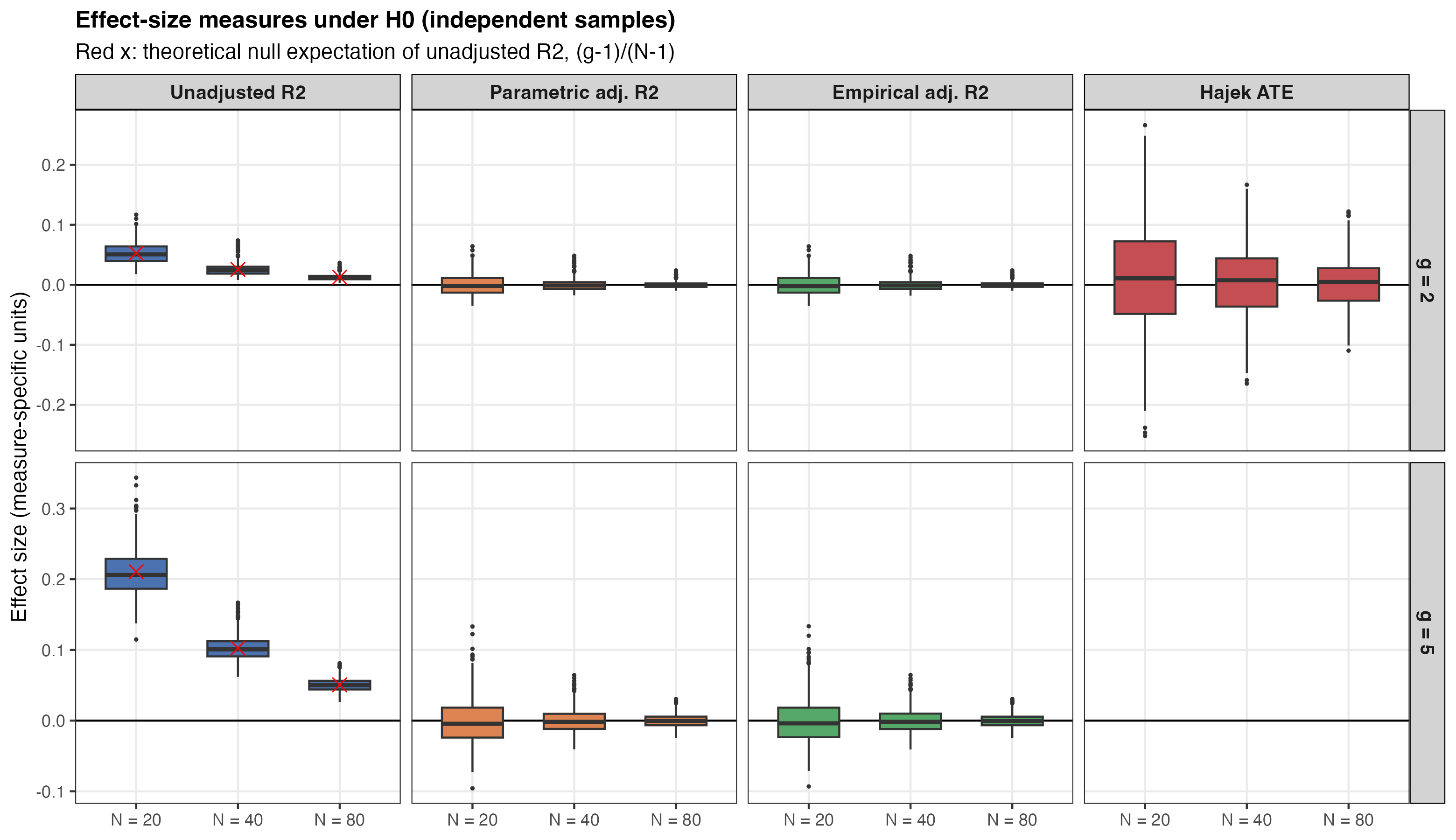}
    \caption{Distributions of four effect-size measures under the null hypothesis
with independent samples (HMP2, one sample per subject, Bray--Curtis
dissimilarity, 500 iterations per setting). Rows show the number of groups
$g$; columns show the effect-size measure. Red crosses mark the theoretical
null expectation of unadjusted $R^2$, $(g-1)/(N-1)$; the horizontal line marks
zero. The H\'ajek estimator is defined only for binary covariates and is
therefore shown for $g = 2$ only.}
\label{fig:effsizeNull}
\end{figure}

Under the null, unadjusted $R^2$ was consistently positive, and its median
closely matched the theoretical expectation $(g-1)/(N-1)$ in every setting
(Fig.~\ref{fig:effsizeNull}). The inflation was largest when the number of
groups was large relative to the sample size: with $g = 5$ and $N = 20$, the
median $R^2$ was approximately 0.21, implying that one fifth of the
community variation was ``explained'' by randomly assigned labels. This
inflation shrank as $N$ increased but remained clearly nonzero at $N = 80$
($\approx 0.05$ for $g = 5$). Notably, the Euclidean-derived expectation
$(g-1)/(N-1)$ also described the null behavior of $R^2$ under Bray--Curtis
dissimilarity. By contrast, the parametric and empirical null-centered
measures were both centered at zero across all sample sizes and numbers of
groups, with variability that decreased with $N$, and the H\'ajek estimator
was likewise centered near zero. The H\'ajek estimates were more variable than
the null-centered $R^2$ measures, reflecting both its different scale (units of
Bray--Curtis dissimilarity rather than a proportion of variance) and the
sample split, which uses only half of the data to estimate the effect.

Note that the choice among the three $R^2$-based measures affects only the
reported magnitude, not the p-value. Unadjusted
$R^2$, the parametric $\tilde{R}^2$, and the pseudo-$F$ statistic are monotone
functions of one another, and the empirical $\tilde{R}^2$ is evaluated against
the same permutation null as $R^2$. Null-centering thus corrects the interpretation of effect
magnitude without changing which associations are declared significant. The
H\'ajek estimator, in contrast, is a distinct statistic with its own permutation
null and p-value. Because half of the samples are reserved for estimating the
control centroid, it is expected to be somewhat less powerful than PERMANOVA at
a given sample size, a trade-off for an
effect size with a direct interpretation on the distance scale. We further show how the H\'ajek
effect size further decomposes into location and dispersion components
(Appendix Figure~\ref{fig:hajekDecomp}) using the nasal data introduced earlier.

\subsection{PERMANOVA p-values are anti-conservative for block-invariant covariates}

In order to compute accurate p-values when applying PERMANOVA to non-independent sample sets (e.g. those with block structure or drawn longitudinally from the same source), we consider both block-invariant and block-variable covariates. Block-invariant covariates remain the same for each sample in a non-independent set, such as a subject-level metadatum (e.g. sex) across a series of repeated measures. We evaluated the ability of different PERMANOVA variants to capture these block-invariant effects using 1) a Kolmogorov-Smirnov (K-S) Test \citep{kolmogorov1933} of empirical p-value deviation from uniform over $[0,1]$ and 2) the corresponding false positive rate (FPR).

For the block-invariant covariate of interest, the p-value distributions from unadjusted PERMANOVA and ``Var+Subj" PERMANOVA adjustment (see Methods) yielded an excess of false positives under the null signal case (non-independence measure = $0.5$), while the p-value distribution from between-block-permuted PERMANOVA and within-block-permuted Hájek estimator correctly recovered the uniform distribution under the null, while remaining sensitive to true effects across all signal strengths. It is of note that our causal estimator is slightly under-powered compared to the other methods, perhaps because of the sample-splitting scheme through which half of the data is used for centroid estimation and the remaining data is used for effect size estimation (Fig. \ref{fig:hmp2_invariantPval}). Furthermore, Fig. \ref{fig:hmp2_invariantVolcano} in Appendix \ref{sec:app:hmp2_volcanos} illustrates how the p-values vary with respect to effect size.

Specifically, the Hájek estimator and Between Block PERMANOVA yielded statistics and p-values of $(0.036, 0.536)$ and $(0.036, 0.536)$, respectively, providing evidence that these methods mitigate confounding bias in the invariant case. At $\alpha = 0.05,$ the FPR of the two permutation-aware methods was less than $6\%,$ while the alternative methods both showed FPRs over $70\%$. Thus, we determine that our permutation scheme is uniquely positioned to provide valid confounding adjustment for block-invariant covariates (Table \ref{tab:hmp2_invariantTab}).

\begin{table}[!h]

\caption{Kolmogorov--Smirnov test for Each Method Block-Invariant Setting}
\centering
\begin{tabular}[t]{c|ccc}
\hline
Method & K-S Statistic & K-S p-value & FPR\\
\hline
Hájek & \textbf{0.036} & \textbf{0.536} & 0.044 \\
Between Block PERMANOVA & \textbf{0.036} & \textbf{0.536} & \textbf{0.054}\\
Unadj. PERMANOVA & 0.768 & 0.000 & 0.738 \\
Var+Subj PERMANOVA & 0.988 & 0.000 & 1.00 \\
\hline
\end{tabular}
\label{tab:hmp2_invariantTab}
\end{table}


\begin{figure*}[h!]
    \centering
\includegraphics[width=\textwidth]{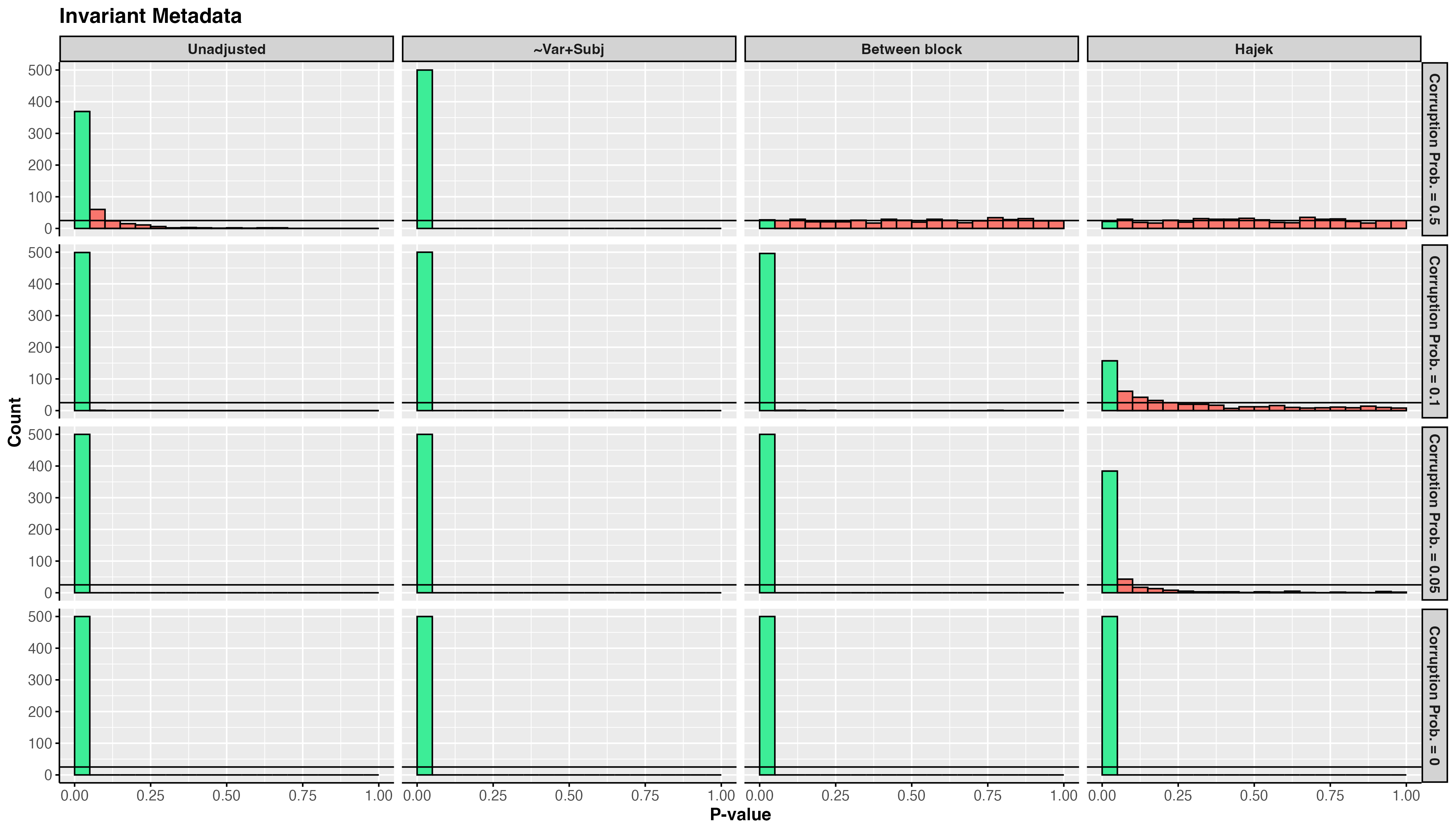}
    \caption{Empirical distribution of p-values for invariant metadata}
    \label{fig:hmp2_invariantPval}
\end{figure*}

\subsection{Permutation schemes must capture block-variant covariates for accurate p-values}

Sample non-independence can also affect the calculation of p-values for block-variant covariates during PERMANOVA, i.e. those that are sample- rather than block-specific; for example, water column depth as measured at each of several sites. For such variant covariates, the empirical distribution of p-values from unadjusted PERMANOVA, $\sim$Subj+Var PERMANOVA, and $\sim$Var+Subj PERMANOVA still yield an excess of false positives under the null signal case (Effect size = $0$). P-values from the remaining four methods - which include the adonis2 permutation implementations (G1 and G2) - correctly recover a uniform distribution under the null while remaining sensitive to true effects across all signal strengths (Fig. \ref{fig:hmp2_variantPval}). Furthermore, Fig. \ref{fig:hmp2_variantVolcano} in Appendix \ref{sec:app:hmp2_volcanos} illustrates how the p-values vary with respect to effect size.

In the block-variant case, these tests demonstrate a stark contrast between the four ``non-independence-aware" methods and the three ``naive" PERMANOVA methods, with the latter all yielding p-value $\approx 0.$ Furthermore, the FPR of $\sim$Var+Subj PERMANOVA and unadjusted PERMANOVA demonstrate these two methods' inadmissibility in the variant covariate setting, with both yielding rates above $30\%.$ On the other hand, the FPR of the four aforementioned methods are upper bounded by $6\%.$ It is clear from these experiments that failure to properly account for confounding in the variant covariate case can yield severely distorted hypothesis tests for effect size testing in PERMANOVA-type problems. While this is methodologically apparent and, as shown here, addressable via
several approaches, it has sometimes remained overlooked in practice
\citep{kleine2021data,zhu2021constraining} (Table~\ref{tab:hmp2_variantTab}).

\begin{table}[!h]

\caption{Kolmogorov--Smirnov Statistics for Each Method Block-Variant Setting}
\centering
\begin{tabular}[t]{c|cccc}
\hline
Method & K-S Statistic & K-S p-value & FPR\\
\hline
Hájek & \textbf{0.029} & \textbf{0.794} & \textbf{0.056}\\
Within Block PERMANOVA & 0.052 & 0.134 & 0.042 \\
G1 PERMANOVA & 0.052 & 0.134 & 0.042 \\
G2 PERMANOVA & 0.054 & 0.108 & 0.036 \\
$\sim$ Subj+Var PERMANOVA & 0.062 & 0.043 & 0.042 \\
$\sim$ Var+Subj PERMANOVA & 0.966 & 0.000 & 0.996 \\
Unadj. PERMANOVA & 0.398 & 0.000 & 0.304 \\
\hline
\end{tabular}
\label{tab:hmp2_variantTab}
\end{table}

\begin{figure*}[h!]
    \centering
\includegraphics[width=\textwidth]{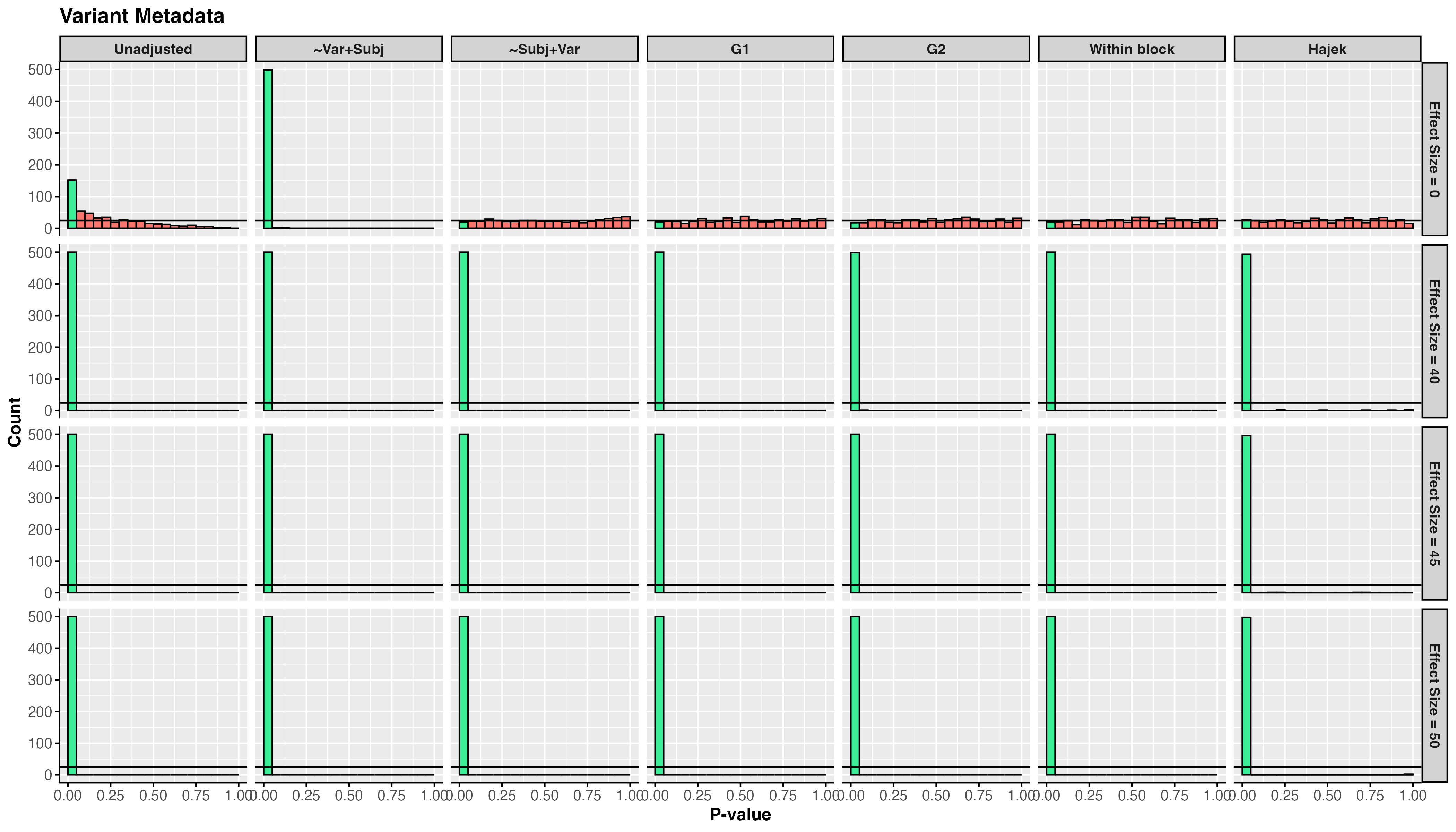}
    \caption{Empirical distribution of p-values for variant metadata}
    \label{fig:hmp2_variantPval}
\end{figure*}

\section{Discussion}\label{sec:discussion}

In this work, we revisited the foundations of effect size estimation for PERMANOVA, proposed a causally-grounded, distance-based alternative, and introduced a permutation scheme under non-independent samples compatible with both PERMANOVA and our new estimator. Motivated by widespread use of PERMANOVA in microbiome and other high-dimensional applications, we first demonstrated that the conventional $R^2$-based measure of explained variation is systematically inflated under the null, particularly when the number of levels in the covariate of interest is large relative to the sample size. In addition to evaluating adjustment methods to account for this, we also constructed an estimator based on the Hájek average treatment effect, defined on sample-wise distances to an estimated control centroid, that (i) admits a causal interpretation under standard identification assumptions and (ii) yields an effect size expressed in the same units as the chosen dissimilarity metric. Additionally, we demonstrated the anti-conservative influence of non-independence among samples during p-value calculation (e.g. time courses or other repeated measures) and developed a permutation scheme that addresses it for both invariant and variant covariates of interest. This generalizes the existing approach for block-variant covariates, and it is exact for single covariates or multiple covariates of the same type simultaneously. All of these features were verified in settings with diverse types and amounts of subject-level confounding, allowing the methods and their implementation to be applied broadly across future studies.


Based on our findings, we offer the following recommendations for practitioners using PERMANOVA or related omnibus tests in high-dimensional settings:

\begin{enumerate}[label=(\roman*)]

    \item \emph{Unadjusted $R^2$ is unreliable as a PERMANOVA effect size in small or many-group designs}. It can be replaced by 1) standard parametric $R^2$ adjustment under Euclidean distance; 2) empirically adjusted $R^2$ (generalizable, at the cost of somewhat increased compute); or 3) the H\'ajek estimator introduced here, which sacrifices power but is expressed in units of the original distance measure.

    \item \emph{PERMANOVA permutations must be adjusted in the presence of non-independence among samples}. This most frequently occurs in blocks, e.g. due to repeated measures. Block-variant covariates can be adjusted by methods that employ within-block permutations, whereas block-invariant covariates must employ between-block permutations.


    \item \emph{Report implementation details.} Methods sections should include the full model formula, the distance metric, the permutation strategy (including any blocking or conditioning), and the definition of the reported effect size. This level of transparency is necessary to assess potential confounding and interpret effect magnitude across studies.
\end{enumerate}


Several potential extensions of these methods could be developed in future work. Our Hájek estimator, for example, currently handles only one binary covariate,
and it measures only the average change in distance from the control group's
centroid. Extensions to multi-level or continuous covariates, multivariate contrasts, and interaction effects would broaden its applicability. 
Second, because we employ sample-splitting,
results vary across splits and only half the subjects inform the effect
estimate. Furthermore, our logistic propensity model can be misspecified, which can bias the estimate.
More advanced cross-fitting or doubly robust approaches could further stabilize the estimator while maintaining valid inference under weaker assumptions.

Additionally, although we considered block-nonindependence in repeated measures, many real-world analyses involve multiple, potentially overlapping sources of non-independence (e.g., batch, study, household, or geographic region). Systematically characterizing the behavior of distance-based effect size estimators under such complex designs, and developing unified permutation frameworks for them, remains an open challenge. In a similar vein, extension of this framework to accommodate covariates of multiple types (invariant and variant) simultaneously using Freedman-Lane permutations would again broaden its utility. Finally, while we focused on PERMANOVA, related distance-based methods such as MiRKAT \citep{zhao2015testing, koh2019distance} and kernel-based association tests may admit analogous effect size constructions; exploring these connections is an important direction for future work.


In summary, we have shown that the conventional $R^2$-based effect size used with PERMANOVA is inherently inflated under the null in many realistic designs and can be difficult to interpret. We introduced both a Hájek-based causal estimator that quantifies the effect of a binary covariate on distances to a control centroid, which yields an effect size with a clear interpretation in the units of the chosen distance metric, and a confounding-aware permutation scheme for PERMANOVA-type methods. This approach provides a principled and transparent alternative for quantifying and testing multivariate effects in longitudinal and other complex high-dimensional studies.

\section{Conflicts of interest}
The authors declare that they have no competing interests.


\bibliographystyle{oup-abbrvnat}
\bibliography{reference}

\section{Funding}
This work is supported in part by the National Institutes of Health under grant number T32HD113301: Artificial Intelligence and Precision Nutrition Training Program, Cornell University. 

\section{Data availability}
The data underlying the main body of this article are publicly available from the Inflammatory Bowel Disease Multi'omics Database (IBDMDB), generated as part of the Integrative Human Microbiome Project \citep{lloyd2019multi}, at \url{https://ibdmdb.org}. We used the merged species-level taxonomic profiles from shotgun metagenomic sequencing (\texttt{taxonomic\_profiles\_3.tsv.gz}, MGX data products release 2018-05-04) and the associated sample metadata (\texttt{hmp2\_metadata\_2018-08-20.csv}). 
Raw sequencing reads are available from the NCBI Sequence Read Archive under BioProject accession PRJNA398089.

Processed nasal taxonomic profiles and metadata were provided by the authors of
\citet{accorsi2020determinants}; raw reads are available from NCBI SRA (BioProject PRJNA610982).

\section{Author contributions statement}

S.Man. conceived and executed the experiments, developed the Hájek estimator, and drafted the manuscript. E.A. contributed to the design of the experiments and development of the permutation scheme. J.N. provided scientific guidance throughout and revised the manuscript. S.Ma provided key conceptual input that helped refine the Hájek estimator throughout its development and revised the manuscript. E.F. provided ongoing scientific guidance, contributed to each part of the methodological development throughout the project, and revised the manuscript. C.H. oversaw the study and provided overall direction, interpretive guidance, methodological guidance, and manuscript revision prior to submission.

\section{Acknowledgments}
 This work was supported in part by funds from the National Institutes of Health [grant number T32HD113301]: Artificial Intelligence and Precision Nutrition Training Program, Cornell University. 

\section{Appendix}
\label{sec:Appendix}

\subsection{$R^2$ equivalence proof and auxiliary figures}
\label{sec:app:R2equivalence}
\begin{proof}
\begin{align*}
&R^2_{\text{PERMANOVA}}
  = \frac{SS_A}{SS_T}\\[0.3em]
  &= \frac{\left(\sum_{i<j} d_{ij}^2/N\right)
          - \left(\sum_{l=1}^g \sum_{i<j} \delta_{ij}^{[l]} d_{ij}^2/n_l\right)}
          {\sum_{i<j} d_{ij}^2/N}\\[0.3em]
  &= \frac{\left(\sum_{i<j} \sum_{v} (Y_{iv} - Y_{jv})^2/N\right)
          - \left(\sum_{l} \sum_{i<j} \delta_{ij}^{[l]}
                  \sum_{v} (Y_{iv} - Y_{jv})^2/n_l\right)}
          {\sum_{i<j}\sum_{v=1}^p (Y_{iv} - Y_{jv})^2/N}\\[0.3em]
  &= \frac{\left(\sum_{v=1}^p \sum_{i=1}^N (Y_{iv} - \overline{Y}_v)^2\right)
          - \left(\sum_{v=1}^p \sum_{l=1}^g \sum_{i=1}^N
                    (Y_{iv} - \overline{Y}_{lv})^2\right)}
          {\sum_{v=1}^p \sum_{i=1}^N (Y_{iv} - \overline{Y}_v)^2}\\[0.3em]
  &= \frac{\sum_{v=1}^p \left(\sum_{i=1}^N (Y_{iv} - \overline{Y}_v)^2
          - \sum_{l=1}^g \sum_{i=1}^N (Y_{iv} - \overline{Y}_{lv})^2\right)}
          {\sum_{v=1}^p \sum_{i=1}^N (Y_{iv} - \overline{Y}_v)^2}\\[0.3em]
  &= \frac{\sum_{v=1}^p \left(\sum_{i=1}^N (Y_{iv} - \overline{Y}_v)^2
          - \sum_{l=1}^g \sum_{i=1}^N (Y_{iv} - \hat{Y}_{lv})^2\right)}
          {\sum_{v=1}^p \sum_{i=1}^N (Y_{iv} - \overline{Y}_v)^2}\\[0.3em]
  &= \frac{SS_{\text{total}} - SS_{\text{res}}}{SS_{\text{total}}}\\
  &= R^2_{\text{LM}}
\end{align*}

where $\hat{Y}$ is the predicted output from a linear model (LM).
\end{proof}

\begin{figure}[h!]
    \centering
\includegraphics[width=\linewidth]{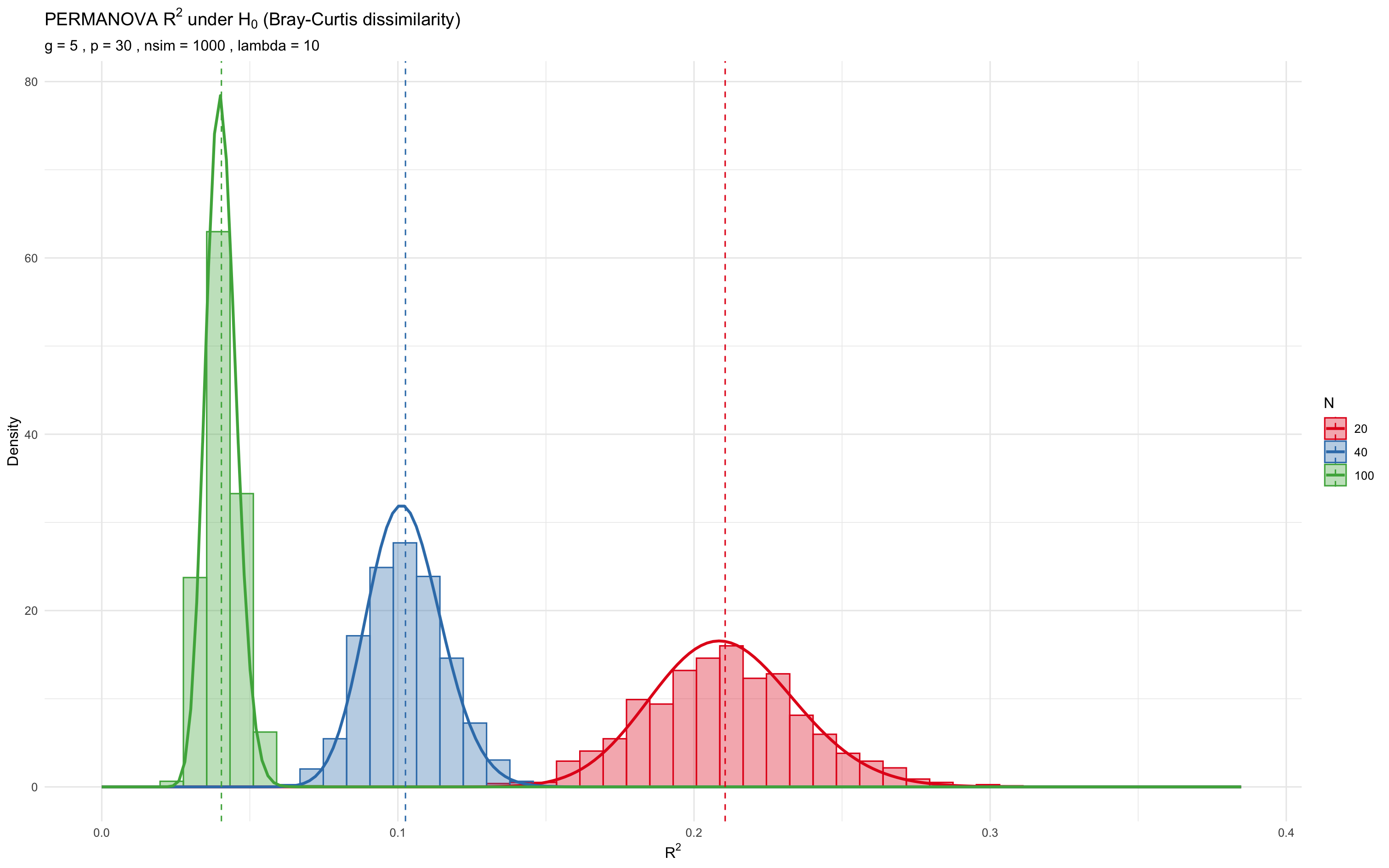}
        \caption{Simulated PERMANOVA $R^2$ distribution under $H_0$, with overlaid $\text{Beta}$ distribution and mean. Red, blue, and green correspond to $N=20$, $N=40$, and $N=100$, respectively. (Bray-Curtis dissimilarity)
}
        \label{fig:betaOverlayBray}
\end{figure}

\begin{figure}
    \centering
    \includegraphics[width=0.7\linewidth]{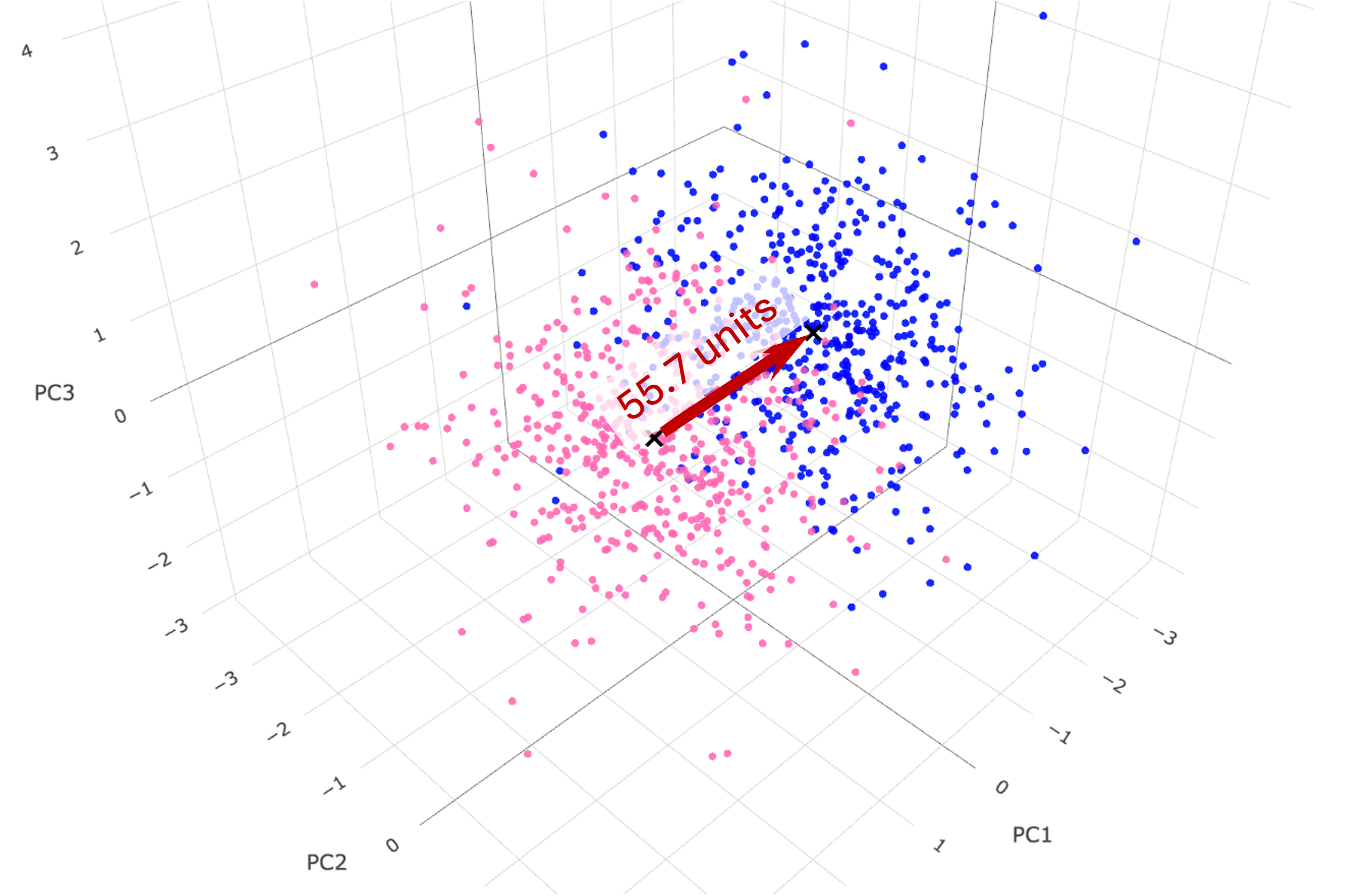}
    \caption{Illustration of Hájek effect size measure}
    \label{fig:hajekDemo}
\end{figure}

\begin{figure}[h]
    \centering
\includegraphics[width=\linewidth]{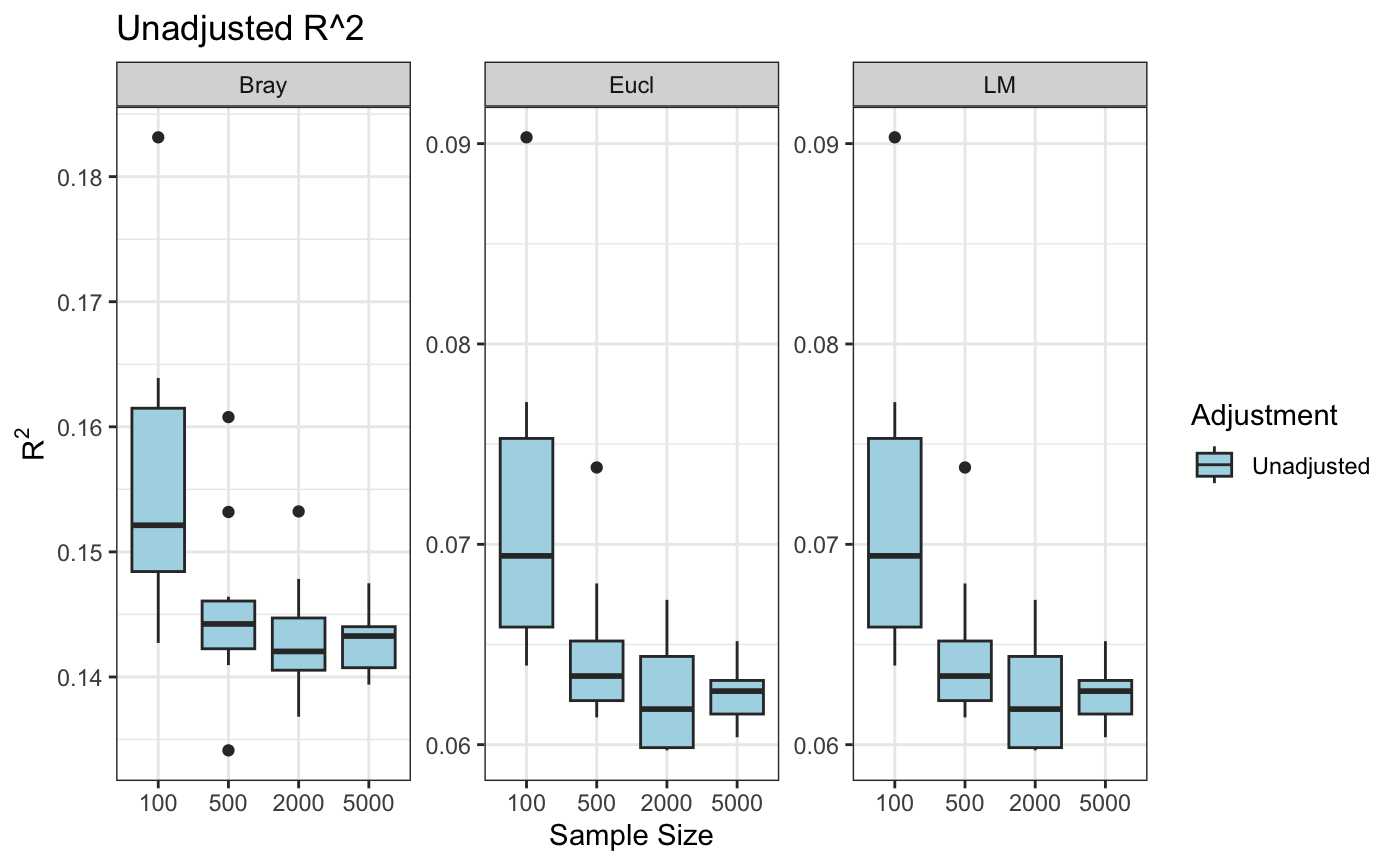}
            \caption{
        Unadjusted $R^2$ values across sample sizes ($n$). 
        Each panel represents a distinct analysis method: from left to right, 
        PERMANOVA with Bray-Curtis dissimilarity, PERMANOVA with Euclidean distance, and linear model (LM). 
        Boxplots display the distribution of $R^2$ values across replicates for each sample size and method.
    }
        \label{fig:sub1}
\end{figure}
\begin{figure}[h]
\includegraphics[width=\linewidth]{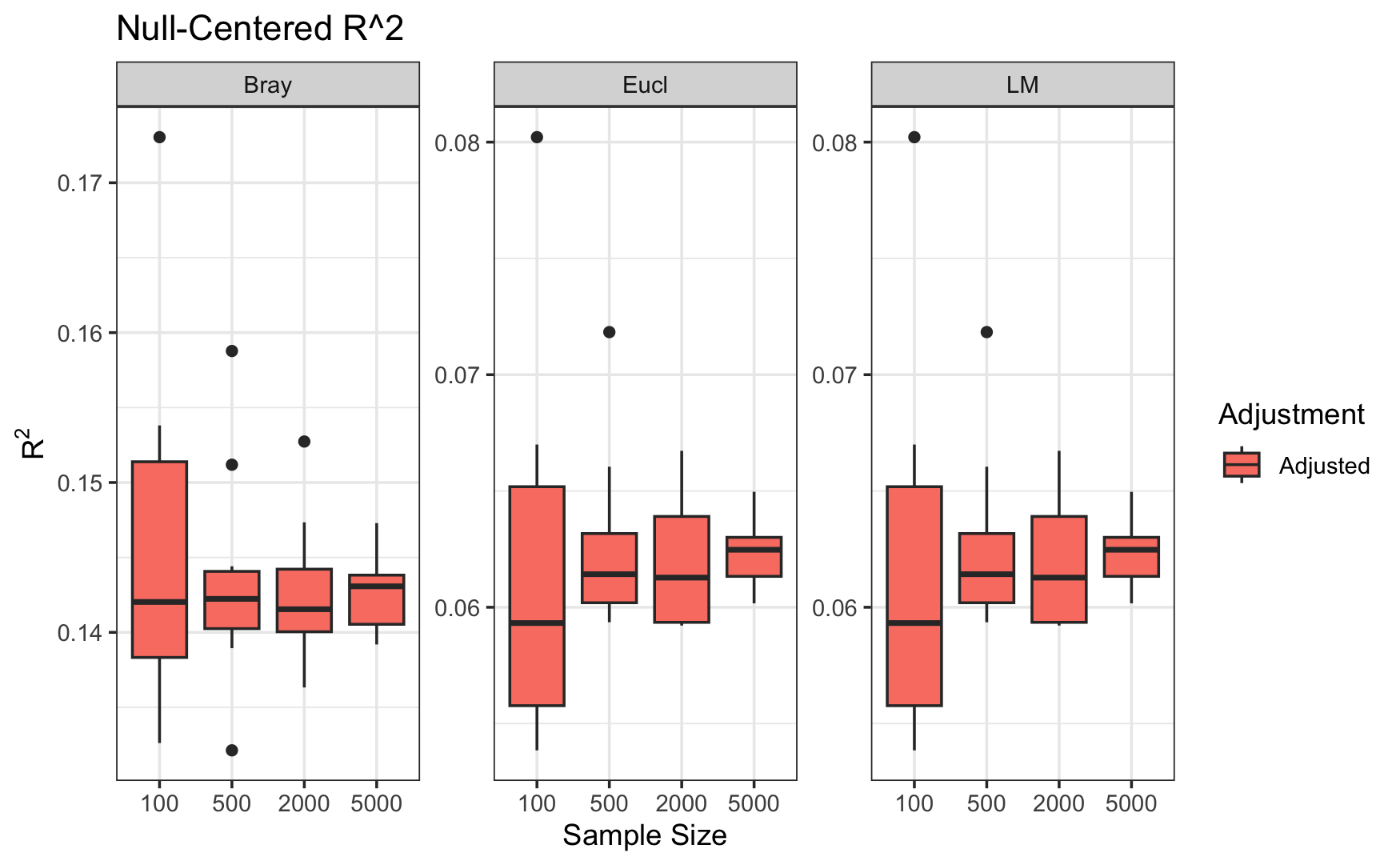}
    \caption{
        Null-centered (B) $R^2$ values across sample sizes ($n$). 
        Each panel represents a distinct analysis method: from left to right, 
        PERMANOVA with Bray-Curtis dissimilarity, PERMANOVA with Euclidean distance, and linear model (LM). 
        Boxplots display the distribution of $R^2$ values across replicates for each sample size and method.
    }
    \label{fig:sub2}
\end{figure}

\begin{figure}
    \centering
    \includegraphics[width=0.5\linewidth]{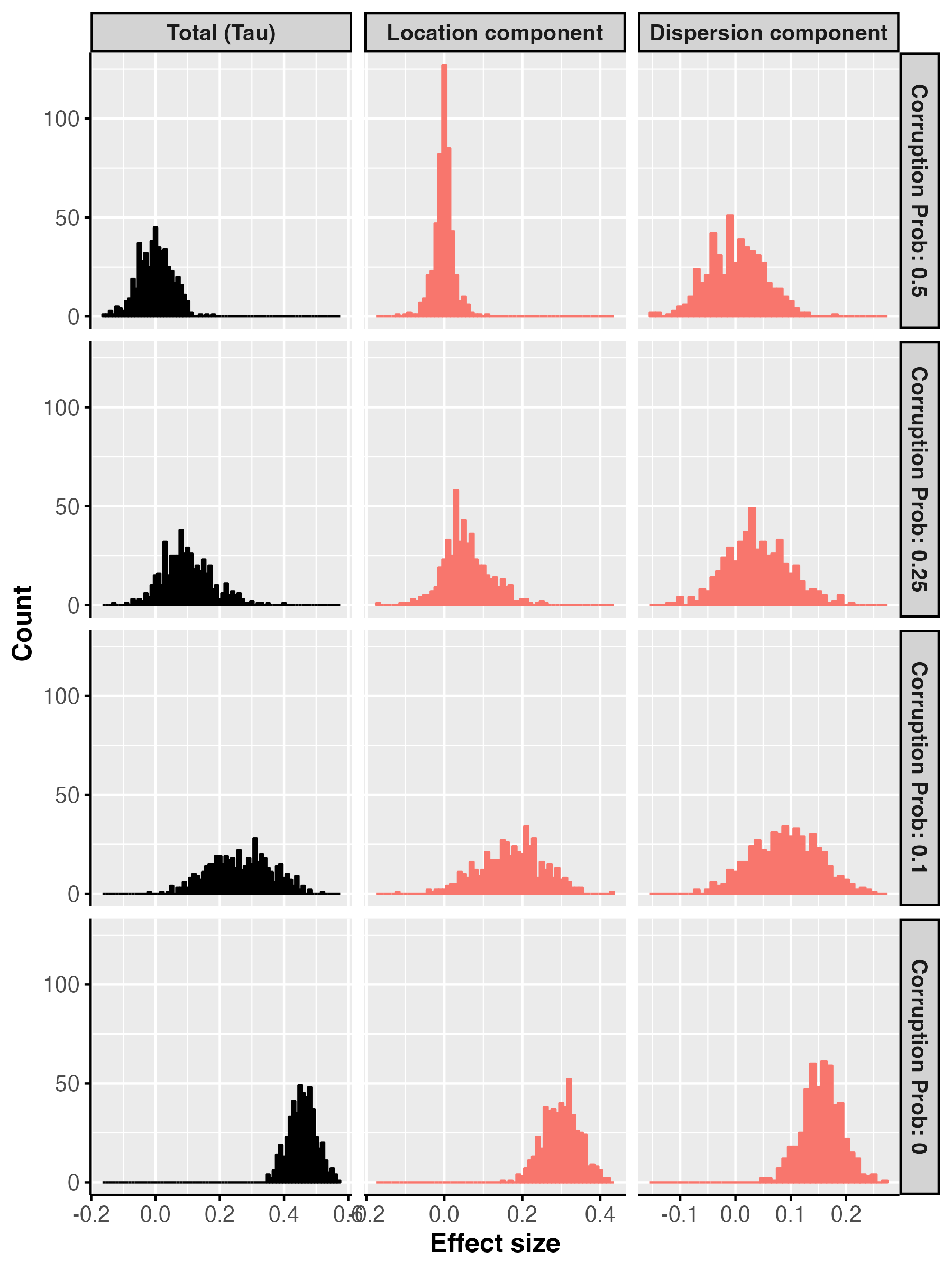}
    \caption{Location-dispersion decomposition of Hájek effect size}
    \label{fig:hajekDecomp}
\end{figure}

\subsection{Validity of Proposed Repeat Measures Permutation in Omnibus Tests}
\label{sec:app:omnibus_validity}

Algorithm~\ref{alg:customInvariant} generates the null distribution by applying
a between-block operator $\mathbf{P}$ that reassigns subject-level covariate
values across subjects, or a within-block operator $\mathbf{Q}$ that shuffles
observation-level covariate values within each subject.
We show here that when all covariates on the right-hand side of the formula are
of the same type---either all invariant (subject-level) or all variant
(observation-level)---the permutation scheme yields an exact omnibus test.
The single-covariate case studied throughout this paper is a special case.
When covariates of mixed types are present, this exactness guarantee does not
apply, and the software implementation of \texttt{permabrioche} guards against this.

\paragraph{Setup.}
Index samples by $(i, j)$: subject $i = 1, \ldots, n$, replicate
$j = 1, \ldots, m_i$, $N = \sum_i m_i$.
Let $\mathbf{W} = [\mathbf{1}, X_1, \dots, X_k]$ be the $N \times q$ model matrix with an intercept and $k$ covariates, all of the same type, coded so that $\mathbf{W}$ has full column rank, and let
$\mathbf{H}_\mathbf{W} = \mathbf{W}(\mathbf{W}^\top\mathbf{W})^{-1}\mathbf{W}^\top$
be its hat matrix.
The omnibus test statistic is
\begin{equation}
\frac{\operatorname{tr}(\mathbf{H_W}\mathbf{G}) / (q-1)}{\operatorname{tr}((\mathbf{I} - \mathbf{H_W})\mathbf{G}) / (N - q)}
  \label{eq:Fomni}
\end{equation}
where $\mathbf{G}$ is Gower's matrix; for a single $g$-level factor, $q = g$ and (2) reduces to (1).
Each permuted replicate forms $\mathbf{W}^*$ by applying a single permutation
to all $k$ columns of $\mathbf{W}$ simultaneously---either a permutation
$\sigma$ of the $n$ subject labels (all-invariant case) or, for each subject
$i$, a permutation $\pi_i$ of the $m_i$ within-subject indices (all-variant
case)---rebuilds $\mathbf{H}_{\mathbf{W}^*}$, and recomputes
$F^{\mathrm{omni}}$. Note that $\mathbf{G}$ is not permuted.

\begin{proposition}[Exactness of the omnibus permutation test]
\label{prop:omnibus_exact}
Suppose all $k$ covariates in $\mathbf{W}$ are of the same type, either
\begin{enumerate}[(i)]
  \item All-invariant: assume subjects are exchangeable under $H_0^{\mathrm{omni}}$
        (between-subject exchangeability).
  \item All-variant: assume observations within each subject are exchangeable
        under $H_0^{\mathrm{omni}}$ (within-subject exchangeability).
\end{enumerate}
Then the omnibus permutation test based on $F^{\mathrm{omni}}$
in~\eqref{eq:Fomni} is exact under $H_0^{\mathrm{omni}}$: the p-value
is uniformly distributed on $\{1/(B+1), 2/(B+1), \ldots, 1\}$,
the set of achievable p-values given $B$ permutations.
\end{proposition}

\begin{proof}
We establish two facts and combine them.

\medskip
\noindent\textit{Fact 1: $\mathbf{G} \perp\!\!\!\perp \mathbf{W}$ under
$H_0^{\mathrm{omni}}$.}
The omnibus null hypothesis asserts that the outcome $Y$ is jointly independent
of the entire covariate design \citep{anderson2001new, anderson2008permanova+}:
\[
  H_0^{\mathrm{omni}}:\quad Y \;\perp\!\!\!\perp\; (X_1,\ldots,X_k).
\]
Since $\mathbf{G}$ is a deterministic function of $Y$ (computed via the chosen
dissimilarity measure, squaring, and double-centring), $\mathbf{G} \perp\!\!\!\perp \mathbf{W}$.
The joint distribution of $(\mathbf{G}, \mathbf{W})$ therefore factors as
$p(\mathbf{G})\cdot p(\mathbf{W})$.

\medskip
\noindent\textit{Fact 2: the permutation preserves $p(\mathbf{W})$.}
Because all covariates are of the same type, the permutation applies a
\emph{single} index vector to all $k$ columns of $\mathbf{W}$ simultaneously.
This is equivalent to permuting the row labels of $\mathbf{W}$.

\begin{itemize}
  \item \emph{All-invariant.} Each column of $\mathbf{W}$ is constant within
    subjects, so $\mathbf{W}$ is determined by the $n$ subject-level row vectors
    $(w_1, \ldots, w_n)$.
    The permutation $\sigma$ reorders these to
    $(w_{\sigma(1)}, \ldots, w_{\sigma(n)})$.
    By between-subject exchangeability, the joint distribution of subject-level
    covariate vectors is invariant to permutation of subject labels, so
    $\mathbf{W}^* \stackrel{d}{=} \mathbf{W}$.

  \item \emph{All-variant.} Within subject $i$, the permutation $\pi_i$
    reorders the $m_i$ rows of that subject's block of $\mathbf{W}$, with the
    same $\pi_i$ applied to all $k$ columns.
    By within-subject exchangeability, the joint distribution of each subject's
    multivariate observation sequence is invariant to such reordering, so
    $\mathbf{W}^* \stackrel{d}{=} \mathbf{W}$.
\end{itemize}

\medskip
\noindent\textit{Conclusion.}
Combining Facts 1 and 2:
\[
  p(\mathbf{G},\,\mathbf{W}^*)
  \;=\; p(\mathbf{G})\cdot p(\mathbf{W}^*)
  \;=\; p(\mathbf{G})\cdot p(\mathbf{W})
  \;=\; p(\mathbf{G},\,\mathbf{W}),
\]
where the first equality uses $\mathbf{G} \perp\!\!\!\perp \mathbf{W}^*$
(since $\mathbf{W}^*$ is a deterministic function of $\mathbf{W}$ and $\mathbf{G}
\perp\!\!\!\perp \mathbf{W}$), and the second uses
$\mathbf{W}^* \stackrel{d}{=} \mathbf{W}$.
Therefore
\[
  F^{\mathrm{omni}}(\mathbf{G},\,\mathbf{W}^*)
  \;\stackrel{d}{=}\;
  F^{\mathrm{omni}}(\mathbf{G},\,\mathbf{W})
  \quad \text{under } H_0^{\mathrm{omni}},
\]
and the permutation p-value is exact up to
discretization.
\end{proof}

\begin{remark}[Single covariate]
The single-covariate setting studied throughout this paper is a special case of
the all-same-type setting with $k = 1$. A single invariant covariate is covered
by case~(i); a single variant covariate by case~(ii).
\end{remark}

\subsubsection{Details for Figures \ref{fig:sub1} and \ref{fig:sub2}}
\label{sec:app:QMD}.

We used real metagenomic data from human stool samples as provided by the \texttt{curatedMetagenomicData} \cite{pasolli2017accessible} R/Bioconductor package. We filtered for control samples from the \texttt{sampleMetadata} data set, restricted to the stool body site. To reduce repeated measures, we retained only the first visit from each subject. We additionally restricted to studies with more than 1000 samples. The corresponding relative abundance profiles were retrieved, filtered to features present in at least 10\% of all samples, and further filtered to samples with total read counts above 90. Feature abundances were subsequently re-normalized to sum to 1 per sample.

\subsubsection{Simulation Design}

We evaluated the properties of unadjusted and null-centered $R^2$ by repeatedly sampling $N \in \{100, 500, 2000, 5000\}$ samples (with replacement) from the dataset, performing 10 independent replicates for each $n$. For each replicate, we computed $R^2$ across several methods:
\begin{itemize}
  \item \text{PERMANOVA} with Bray-Curtis dissimilarity (\texttt{vegan::adonis2})
  \item \text{PERMANOVA} with Euclidean distance (\texttt{vegan::adonis2})
  \item \text{Linear model} (LM) with study as a categorical covariate, applied to log relative abundances
\end{itemize}

\subsubsection{$R^2$ and null-centered $R^2$ Calculation}

For each sampled dataset, we computed both the standard (unadjusted) $R^2$ and the adjusted $R^2$:
\[
R^2 = 1 - \frac{\mathrm{RSS}}{\mathrm{TSS}}, \qquad
\tilde R^2 = R^2 - \frac{g-1}{N-1},
\]
where $N \in \{100, 500, 2000, 5000\}$, and $g = 2.$

\subsubsection{Relation to $\eta^2$ and $\omega^2$} \label{sec:app:omega2}
The $R^2$ values reported for PERMANOVA correspond conceptually to $\eta^2$-type effect sizes, representing the proportion of variation in the distance matrix attributable to each factor. While $\eta^2$ is known to be biased upward in finite samples and bias‑corrected estimators such as $\omega^2$ have been proposed for univariate models, their extension to distance‑based permutation frameworks is not standardized. Accordingly, $R^2$ values are interpreted here as descriptive, model‑dependent measures of effect magnitude rather than unbiased estimates of population variance explained by the covariate.

\subsubsection{Location--Dispersion Decomposition}

\label{sec:app:locationDispersion}

\begin{proof}
Let $X_i \in \mathbb{R}^p$ denote the feature vector for an observation in the test set $(i\in \mathcal{I}_{test})$. 
Define the Euclidean outcome
\[
Y_i = \|X_i -  \hat \mu_{0,\mathrm{train}}\|_2^2.
\]
For $t \in \{0,1\}$, let
\[
\mu_{t,\mathrm{test}} = \mathbb{E}[X \mid T=t],
\qquad
\Sigma_t = \operatorname{Cov}(X \mid T=t)
\]
denote the conditional mean and covariance of $X$ in the test population.

For any $x, \mu, c \in \mathbb{R}^p$, the following identity holds:
\begin{align*}
\|x - c\|_2^2
&= \|(x - \mu) + (\mu - c)\|_2^2\\
&= \sum_{i=1}^p (x_i - \mu_i + \mu_i - c_i)^2\\
&= \sum_{i=1}^p(x_i-\mu_i)^2 + (\mu_i - c_i)^2 + 2(x_i - \mu_i)(\mu_i - c_i)\\
&= \|x - \mu\|_2^2 + \|\mu - c\|_2^2 + 2(x-\mu)^\top(\mu-c).
\end{align*}

Taking conditional expectation given $T=t$, with $\mu=\mu_{t,\mathrm{test}}$, yields
\begin{align*}
\mathbb{E}\!\left[\|X - c\|_2^2 \mid T=t\right]
&= \mathbb{E}\!\left[\|X - \mu_{t,\mathrm{test}}\|_2^2 \mid T=t\right]
+ \|\mu_{t,\mathrm{test}} - c\|_2^2 \\
&\quad + 2\,\mathbb{E}\!\left[(X-\mu_{t,\mathrm{test}})^\top(\mu_{t,\mathrm{test}}-c)\mid T=t\right].
\end{align*}
The cross term vanishes because $\mathbb{E}[X-\mu_{t,\mathrm{test}}\mid T=t]=0$. Moreover,
\[
\mathbb{E}\!\left[\|X - \mu_{t,\mathrm{test}}\|_2^2 \mid T=t\right]
= \operatorname{tr}\!\left(\operatorname{Cov}(X\mid T=t)\right)
= \operatorname{tr}(\Sigma_t).
\]
Thus,
\[
\mathbb{E}\!\left[\|X - c\|_2^2 \mid T=t\right]
= \|\mu_{t,\mathrm{test}} - c\|_2^2 + \operatorname{tr}(\Sigma_t).
\]

Subtracting the expressions for $t=0$ and $t=1$ gives
\[
\mathbb{E}[Y\mid T=1] - \mathbb{E}[Y\mid T=0]
=
\big(\|\mu_{1,\mathrm{test}} - c\|_2^2 - \|\mu_{0,\mathrm{test}} - c\|_2^2\big)
+
\big(\operatorname{tr}(\Sigma_1) - \operatorname{tr}(\Sigma_0)\big).
\]
This yields the decomposition
\[
\tau = \tau_{\mathrm{loc}} + \tau_{\mathrm{disp}},
\]
where $\tau_{\mathrm{loc}}$ captures differences in location (centroid shift) and
$\tau_{\mathrm{disp}}$ captures differences in dispersion (within‑group variability).
\end{proof}

\textbf{Extension to estimated propensity scores.}
The argument above uses only the fact that the cross term vanishes when $X$ is
centred at its own mean. It therefore holds for any normalized weights. Let
$w_i^{(1)} \propto T_i/\hat e_i$ and $w_i^{(0)} \propto (1-T_i)/(1-\hat e_i)$ be the
H\'ajek weights, normalized to sum to one within each arm. Replacing the
conditional expectations with weighted means gives the weighted centroid
$\tilde\mu_t = \sum_i w_i^{(t)} X_i$ and the weighted total variance
$\operatorname{tr}(\tilde\Sigma_t) = \sum_i w_i^{(t)}\lVert X_i - \tilde\mu_t\rVert_2^2$.
Because $\sum_i w_i^{(t)}(X_i - \tilde\mu_t) = 0$, the same steps give
\[
\hat\tau
= \lVert\tilde\mu_1 - c\rVert_2^2 - \lVert\tilde\mu_0 - c\rVert_2^2
+ \operatorname{tr}(\tilde\Sigma_1) - \operatorname{tr}(\tilde\Sigma_0),
\]
so the decomposition holds exactly for the H\'ajek estimate under any propensity
model. With a constant propensity the weights are uniform, and this reduces to
the unweighted case above.

\subsubsection{Implementation Notes}
\begin{itemize}
\item All analyses used log-transformed relative abundances with pseudocount $1\times10^{-5}$ to avoid $\log(0)$.
\item All results are summarized as distributions of $R^2$ values over simulation replicates.
\item Number of permutations $999$ throughout.
\end{itemize}


\subsection{Volcano Plots}
\label{sec:app:hmp2_volcanos}
\begin{figure*}[h!]
    \centering
\includegraphics[width=\textwidth]{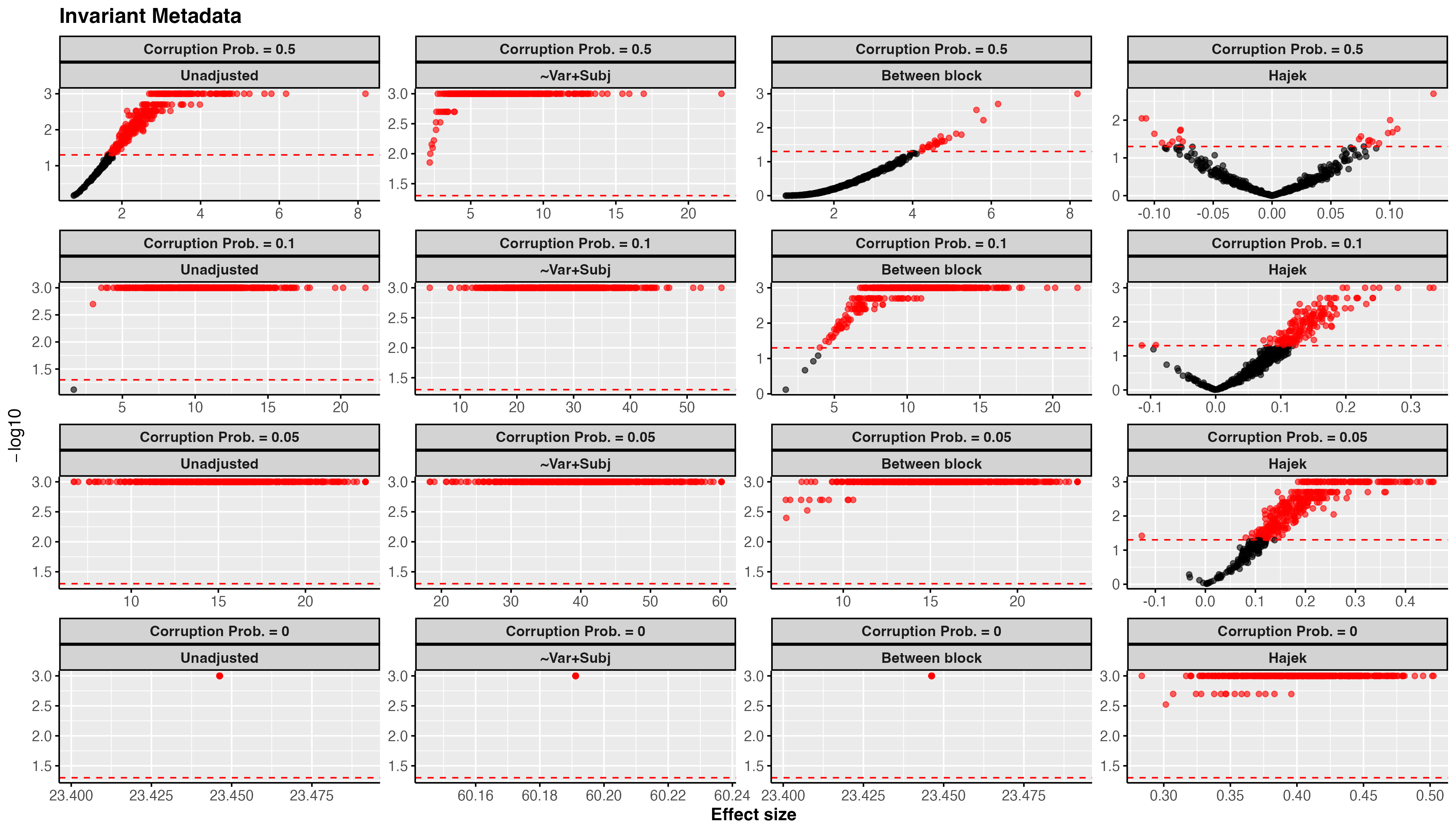}
    \caption{Volcano plot for Invariant Metadata}
    \label{fig:hmp2_invariantVolcano}
\end{figure*}

\begin{figure*}[h!]
    \centering
\includegraphics[width=\textwidth]{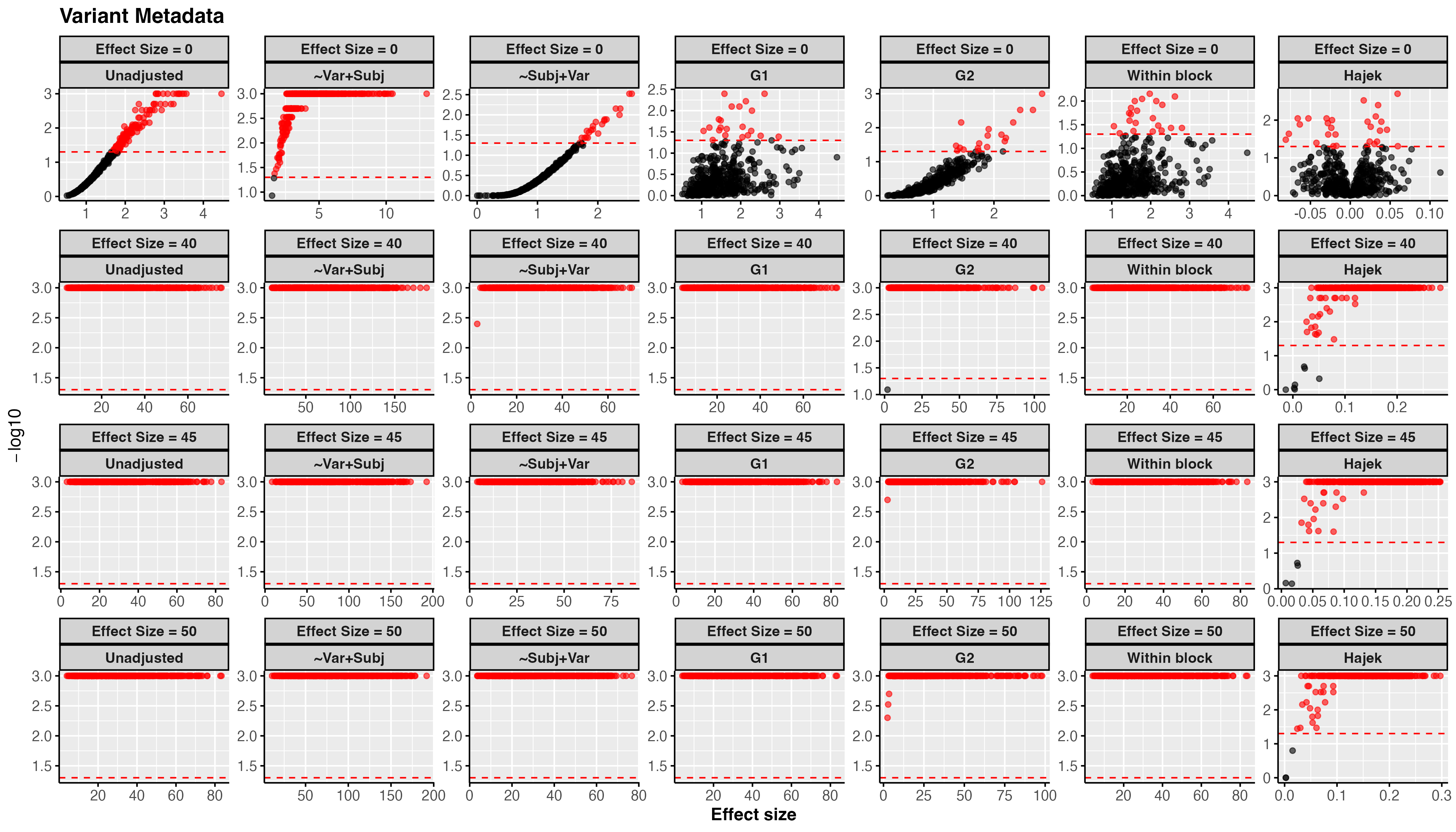}
    \caption{Volcano plot for Variant Metadata}
    \label{fig:hmp2_variantVolcano}
\end{figure*}

\subsection{Nasal Microbiome Data Analysis}\label{sec:app:staph}
\subsubsection{Invariant Metadata}

\begin{figure*}[h!]
    \centering
\includegraphics[width=\textwidth]{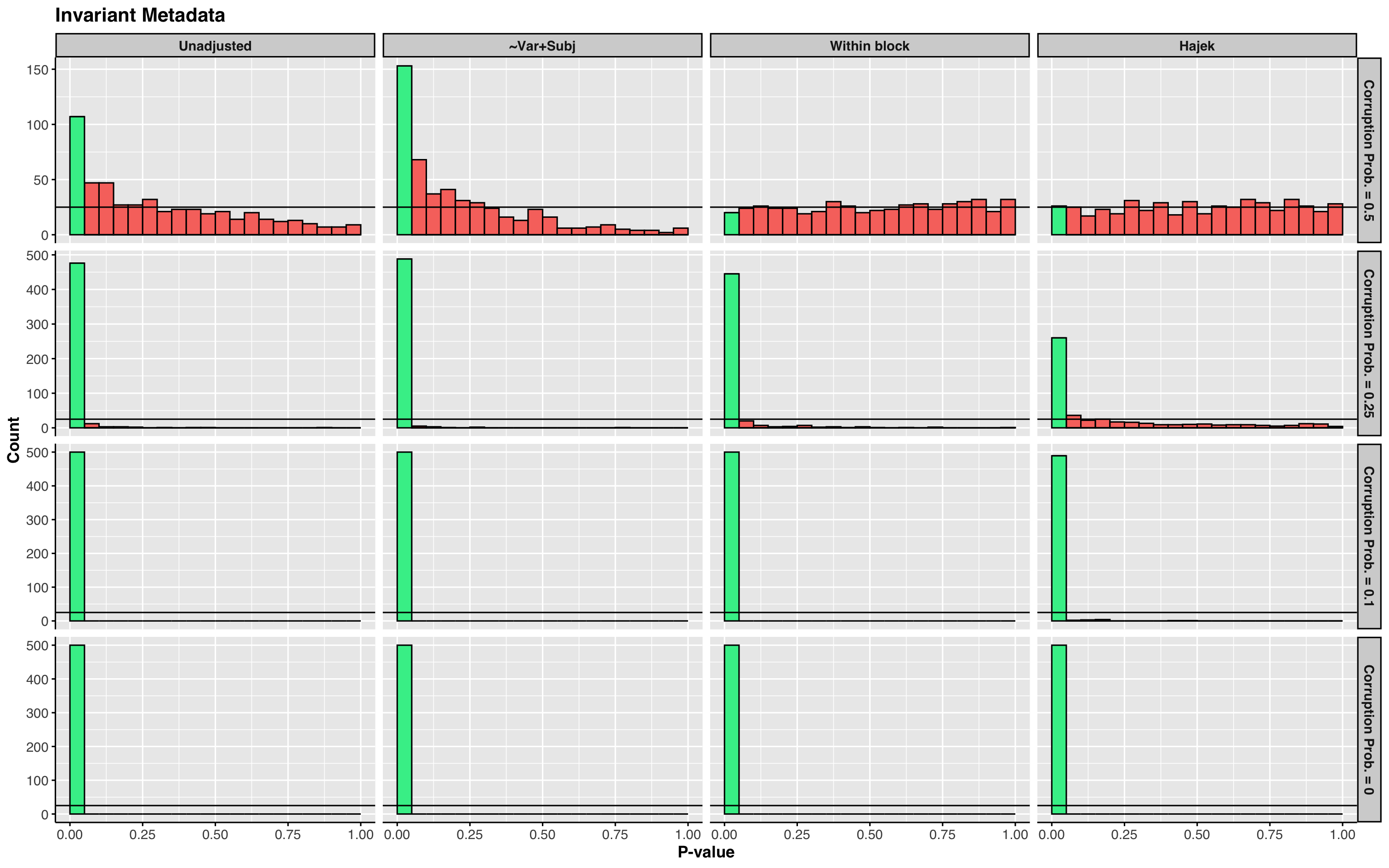}
    \caption{Empirical distribution of p-values for invariant metadata}
    \label{fig:invariantPval}
\end{figure*}

\begin{table}[!h]

\caption{Kolmogorov--Smirnov test for Each Method - Invariant Setting}
\centering
\begin{tabular}[t]{c|ccc}
\hline
Method & K-S Statistic & K-S p-value & FPR\\
\hline
Hájek & \textbf{0.040} & \textbf{0.400} & \textbf{0.052} \\
Between Block PERMANOVA & 0.050 & 0.164 & 0.040\\
Unadj. PERMANOVA & 0.276 & 0.000 & 0.214 \\
Var+Subj PERMANOVA & 0.423 & 0.000 & 0.306 \\
\hline
\end{tabular}
\label{tab:invariantTab}
\end{table}

\subsubsection{Variant Metadata}

\begin{figure*}[h!]
    \centering
\includegraphics[width=\textwidth]{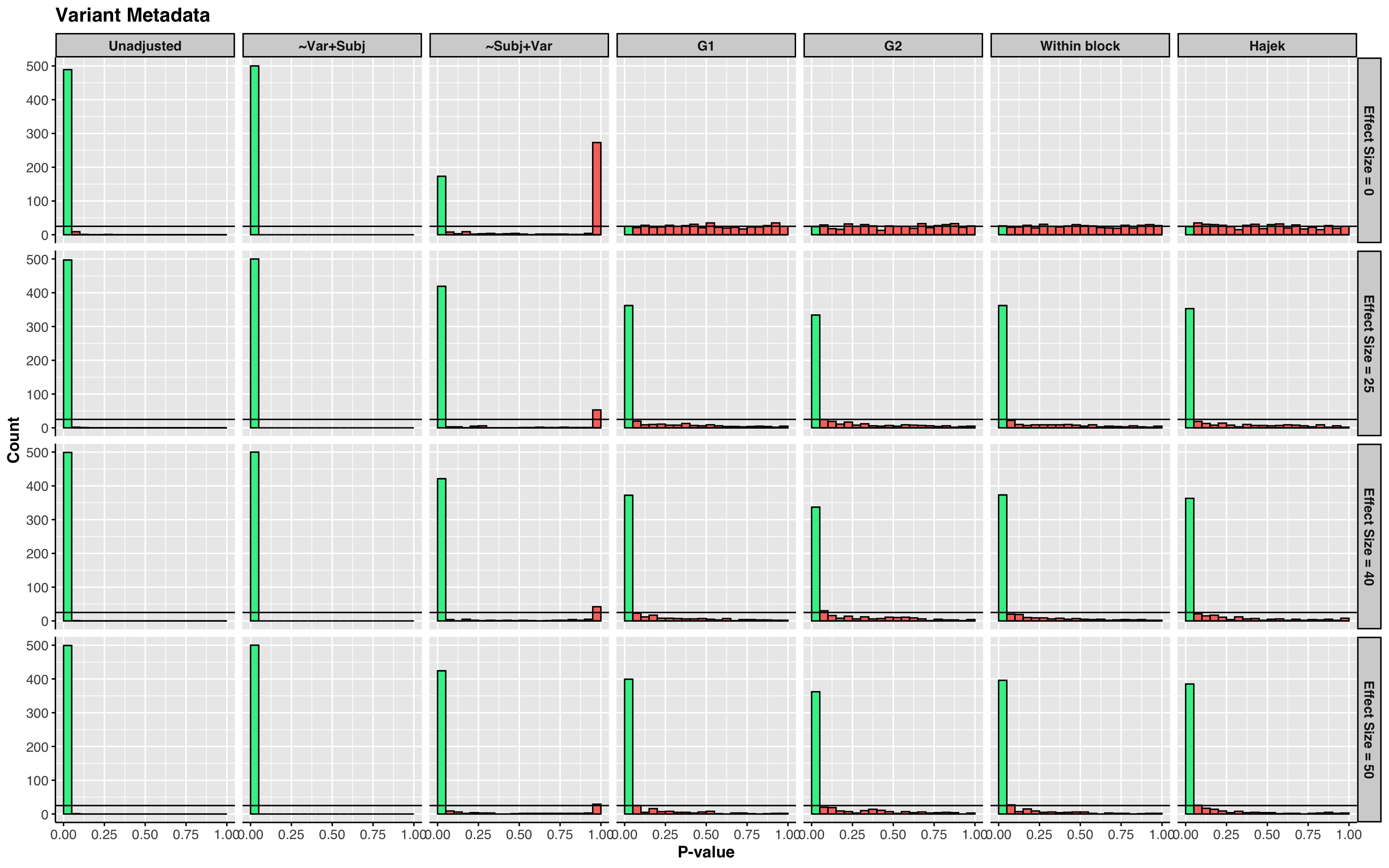}
    \caption{Empirical distribution of p-values for variant metadata}
    \label{fig:variantPval}
\end{figure*}

\begin{table}[!h]

\caption{Kolmogorov--Smirnov Statistics for Each Method - Variant Setting}
\centering
\begin{tabular}[t]{c|cccc}
\hline
Method & K-S Statistic & K-S p-value & FPR\\
\hline
Hájek & 0.057 & 0.078 & 0.048\\
Within Block PERMANOVA & \textbf{0.028} & \textbf{0.828} & 0.052 \\
G1 PERMANOVA & 0.031 & 0.723 & \textbf{0.050} \\
G2 PERMANOVA & 0.035 & 0.573 & 0.048 \\
$\sim$ Subj+Var PERMANOVA & 0.526 & 0.000 & 0.346 \\
$\sim$ Var+Subj PERMANOVA & 0.999 & 0.000 & 1.000 \\
Unadj. PERMANOVA & 0.933 & 0.000 & 0.978 \\
\hline
\end{tabular}
\label{tab:variantTab}
\end{table}

\subsubsection{Volcano Plots}
\label{sec:app:volcanos}
\begin{figure*}[h!]
    \centering
\includegraphics[width=\textwidth]{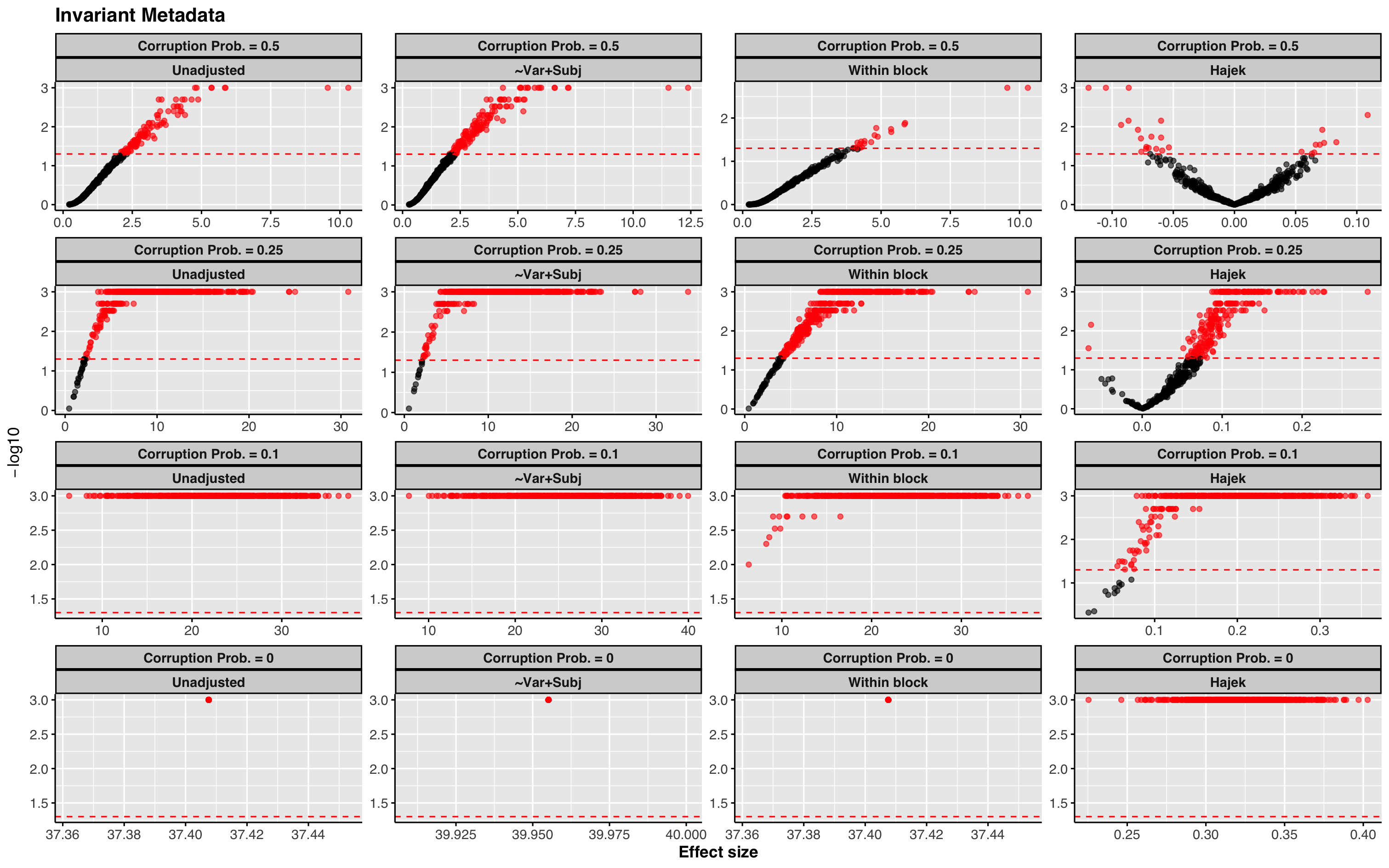}
    \caption{Volcano plot for Invariant Metadata}
    \label{fig:invariantVolcano}
\end{figure*}

\begin{figure*}[h!]
    \centering
\includegraphics[width=\textwidth]{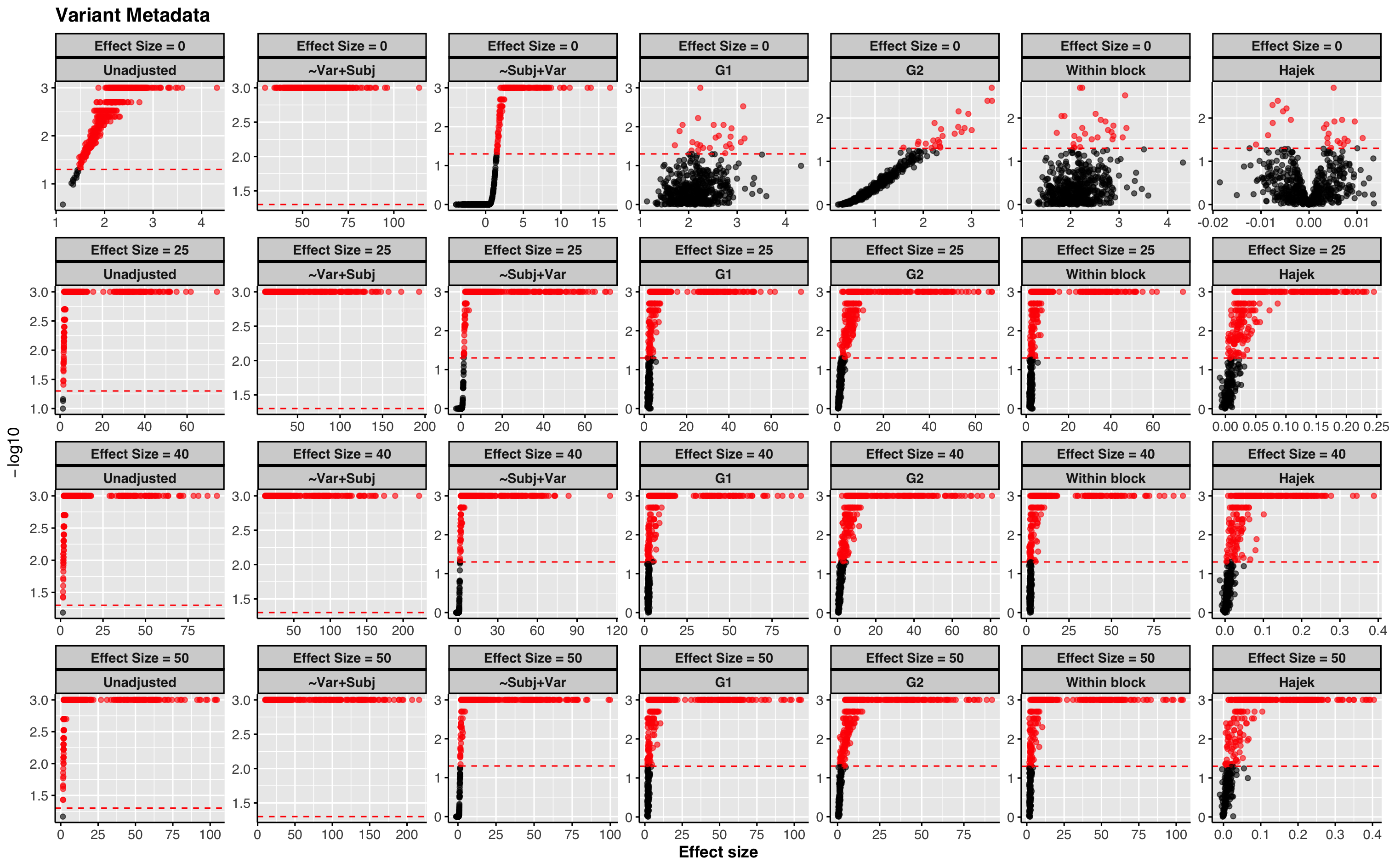}
    \caption{Volcano plot for Variant Metadata}
    \label{fig:variantVolcano}
\end{figure*}



\end{document}